\documentclass[final,1p,times,authoryear]{elsarticle}

\usepackage[activeacute,english]{babel} 
\usepackage[utf8]{inputenc}
\usepackage{amsmath}
\usepackage{amsthm}
\usepackage{amssymb}
\usepackage{amsfonts}
\usepackage{mathdots}
\usepackage{hyperref}
\usepackage{xfrac}

\usepackage[capitalize]{cleveref}
\usepackage{mathtools}
\usepackage{verbatim} 

\DeclarePairedDelimiter{\abs}{\lvert}{\rvert}

\usepackage{url}
\usepackage{enumitem}
\usepackage{todonotes}
\usepackage{xspace}
\usepackage[makeroom]{cancel}

\usepackage[ruled,linesnumbered,boxed,boxruled,vlined]{algorithm2e}
\makeatletter
\let\original@algocf@latexcaption\algocf@latexcaption
\long\def\algocf@latexcaption#1[#2]{%
  \@ifundefined{NR@gettitle}{%
    \def\@currentlabelname{#2}%
  }{%
    \NR@gettitle{#2}%
  }%
  \original@algocf@latexcaption{#1}[{#2}]%
}
\makeatother

\usepackage{color}
\definecolor{red}{rgb}{.7,0,0}
\definecolor{blue}{rgb}{0,0,1}

\def\mcB{\mathcal{B}}

\def\mcG{\mathcal{G}}
\def\mcH{\mathcal{H}}

\def\mcP{\mathcal{P}}

\def\mcS{\mathcal{S}}

\def\bbD{\mathbb{D}}
\def\bbE{\mathbb{E}}

\def\bbR{\mathbb{R}}

\def\bbN{\mathbb{N}}

\def\bbC{\mathbb{C}}
\def\bbP{\mathbb{P}}
\def\bbS{\mathbb{S}}
\usepackage{ifpdf}
\ifpdf
\hypersetup{
 pdftoolbar=true,
 pdfmenubar=true,
 pdfstartview={FitV},
 pdfcreator={overleaf},
 pdfproducer={overleaf},
 linkcolor=red     
 citecolor=green   
 urlcolor=cyan     
 filecolor=magenta 
}
\fi
\usepackage{bookmark}

\makeatletter
\def\th@plain{%
  \thm@notefont{}
  \slshape 
}
\def\th@definition{%
  \thm@notefont{}
  \normalfont 
}
\makeatother
\newtheorem{lem}{Lemma}[section]
\newtheorem{prop}[lem]{Proposition}
\newtheorem{theo}[lem]{Theorem}
\newtheorem{cor}[lem]{Corollary}

\newdefinition{defi}[lem]{Definition}
\newdefinition{problem}{Open problem}
\newdefinition{exam}[lem]{Example}
\newdefinition{obs}{Observation}[lem]
\newdefinition{remark}[lem]{Remark}
\newdefinition{notation}[lem]{Notation}
\newcommand{\eproof}{\hfill\qed}

\def\Hd{\mcH_{n,d}}
\def\Pd{\mcP_{n,d}}
\def\Pdone{\mcP_{1,d}}

\def\Oh{\mathcal{O}}
\def\Tg{\mathrm{T}}
\DeclareMathOperator{\Diff}{D}
\DeclareMathOperator{\diff}{D}

\def\enumber{\mathrm{e}}
\DeclareMathOperator{\vol}{\mathrm{vol}}
\newcommand{\eps}{\varepsilon}

\newcommand{\PP}{\mathbb{P}}

\DeclareMathOperator{\cond}{\texttt{C}_1}

\def\dist{\mathrm{dist}}
\def\fkf{\mathfrak{f}}

\def\fkx{\mathfrak{x}}
\def\fky{\mathfrak{y}}

\def\fkc{\mathfrak{c}}

\DeclarePairedDelimiter\Onorm{\lVert}{\rVert_{1}}

\usepackage{multicol}
\newcommand{\descartes}{\textsc{Descartes}\xspace}

\newcommand{\JSalg}{\textsc{JindalSagraloff}\xspace}

\usepackage{subcaption}

\usepackage{caption}

\newfont{\subsecit}{ptmbi8t at 11pt}  %
\usepackage{showlabels}

\begin{document}

\begin{frontmatter}

  \title{Condition Numbers for the Cube.\\I: Univariate Polynomials and Hypersurfaces}

  \author{Josu\'e Tonelli-Cueto}
  \address{Inria Paris \& IMJ-PRG, Sorbonne Universit\'e \\
    4 place Jussieu, F-75005, Paris, France}
  \ead{josue.tonelli.cueto@bizkaia.eu}
  \ead[https://tonellicueto.xyz]{https://tonellicueto.xyz}

  \author{Elias Tsigaridas}
  \address{Inria Paris \& IMJ-PRG, Sorbonne Universit\'e \\
      4 place Jussieu, F-75005, Paris, France}
  \ead{elias.tsigaridas@inria.fr}
  \ead[https://who.paris.inria.fr/Elias.Tsigaridas/]{https://who.paris.inria.fr/Elias.Tsigaridas/}

  \begin{abstract}
The condition-based complexity analysis framework is one of the gems   of modern numerical algebraic geometry and theoretical computer science. Amond the challenges that it poses is to expand the currently limited range of random polynomials that we can handle. Despite important recent progress, the available tools cannot handle random sparse polynomials and Gaussian polynomials, that is polynomials whose coefficients are i.i.d. Gaussian random variables.

We initiate a condition-based complexity framework based on the norm of the cube that is a step in this direction. We present this framework for real hypersurfaces and univariate polynomials. We demonstrate its capabilities in two problems, under very mild probabilistic assumptions. On the one hand, we show that the average run-time of the Plantinga-Vegter algorithm is polynomial in the degree for random sparse (alas a restricted sparseness structure) polynomials and random Gaussian polynomials. On the other hand, we study the size of the subdivision tree for Descartes' solver and run-time of the solver by \cite{jindalsagraloff2017}. In both cases, we provide a bound that is polynomial in the size of the input (size of the support plus the logarithm of the degree) not only for the average but also for all higher moments.

  \end{abstract}

  \begin{keyword}
    condition number;
    random polynomial;
    subdivision algorithm;
    univariate solver;
  \end{keyword}
  
\end{frontmatter}

\section{Introduction}

The complexity of numerical algorithms is not uniform. It depends on a measure
of the numerical sensitivity of the output with respect to perturbations of the
input, called \emph{condition number} and introduced originally
by~\cite{turing1948} and~\cite{vonneumanngoldstine1947}. If the condition
number of the input is large, then  small numerical perturbations
of the input can significantly change the solution of the problem at hand.
Consequently, numerical algorithms need to use more computational resources to
guarantee a correct computation.

The above phenomenon motivates the condition-based complexity analysis of numerical algorithms.
Although a condition-based complexity analysis can explain the success of an algorithm for a given input, it cannot explain, at least on its own, why a numerical algorithm is efficient. The reason is that condition-based complexity analyses are not input-independent. Thus a common technique that goes back to~\cite{goldstinevonneumann1951},
\cite{demmel1987,demmel1988}, and \cite{smale1997} is to randomize the input. In this way, we obtain a probabilistic complexity analysis that can explain the successful behaviour of an algorithm. Moreover, the framework of smoothed analysis \citep{ST:02}  fully explains the practical behaviour of an algorithm.
We refer the reader to~\citep{Condition} and references therein for more details about this complexity paradigm for numerical algorithms.

After the complete solution\footnote{One should notice that the solution of \cite{lairez2017} does not construct a good starting system, but exploits the randomness of the input as a source of randomness for a deterministic algorithm. Constructing such initial systems is hard, although there was a breakthrough construction for the univariate case by~\cite{etayo2020sequence} \citep[see also][]{beltranlizarte2020}.} of Smale's 17th problem by~\cite{lairez2017}, following the steps of \cite{beltranpardo2008} and \cite{burgissercucker2011}, the main challenge in numerical algebraic geometry is to extend (and analyze) the current algorithms for solving polynomial systems to handle more general inputs; for example sparse and structured polynomials.

In the complex setting, \cite{malajovich2019,malajovich2020arXiv} and \cite{malajovichrojas2002,malajovichrojas2004} did groundbreaking work in the development for numerical algorithms for finding a solution of sparse polynomial systems and recently~\cite{rigid2} (following the ideas of \cite{rigid1}) introduces efficient numerical algorithms for finding a solution of determinant polynomial systems. Additionally, \cite{armentanobeltran2019} provided a probabilistic estimate of the condition number for the Polynomial Eigenvalue Problem, although they did not provide an algorithm.

In the real setting, the situation is far more difficult. For example, as of today, the real version of Smale's 17th problem (asking to decide numerically the feasibility of a real polynomial system) remains open as no algorithm running in finite expected time is known; see~\cite[Ch.~17 and P.18]{Condition} for more details. Of course, ideally, we want to solve the real sparse Smale's 17th problem (which was proposed by \cite{rojasye2005}):
\begin{quote}
    Find an algorithm that finds all the real roots of a real fewnomial system with average run-time bounded by a polynomial in the size of the system (which is the size of the support and the logarithm of the degree).
\end{quote}
However, such a result, although motivated by \cite{khovanskiibook} and
Kushnirenko's hypothesis, seems to be out of reach today. Nevertheless, some
progress has been made by \cite{rojas2020counting} (but not in the
numerical setting). Regarding structured systems, there are some results by
\cite{beltrankozhasov2019} and \citet*{EPR18,EPR19}.

A common problem with many of the current techniques is that they rely on
uni\-ta\-ry/or\-tho\-go\-nal invariance. Therefore, it is central for an effective algorithmic framework 
to develop techniques that do not rely on this invariance to compute with sparse/structured polynomials and more general
probability distributions. We make one step in this research direction by
developing a condition-based complexity framework that relies on the
$\infty$-norm of the cube and which consequently does not rely on unitary
invariance.

We develop the above framework for univariate polynomials and hypersurfaces. We hope to extend it for polynomial systems in a future work. To illustrate its advantages  we apply it to two problems. 
First, to study the complexity of the Plantinga-Vegter algorithm~\citep{plantingavegter2004,burr2017}.
Then, to study the separation bounds of the roots of real univariate polynomials. Using the latter bounds 
we deduce a bound on the average number of subdivisions that Descartes' solver performs to isolate the real roots of univariate polynomials
and we estimate the average bit complexity
of the algorithm by \cite{jindalsagraloff2017} for solving sparse univariate polynomials. The latter bound is polynomial in the input size, providing a first approximation to the real sparse version of Smale's 17th problem by \cite{rojasye2005}.

In the case of the Plantinga-Vegter algorithm, we demonstrate its efficiency by showing that its complexity is polynomial on the average,
for a wide class of random sparse polynomials (Theorem~\ref{thm:mainPVbound}). This significantly extends the results by~\cite{CETC-PV}, \citep[cf.][]{CETC-PVjournal}.
Additionally, our approach applies to Gaussian random polynomials, when all coefficients have the same variance.

We note that our aim is not to show that the Plantinga-Vegter is the most efficient algorithm for random sparse polynomials, but that it remains efficient when we restrict it to a wide class of random sparse polynomials. We note that our bounds depend polynomially on the degree and not logarithmically. A similar approach was employed by~\cite{EPR19} for the algorithm by~\cite{CKMW1} for finding the real zeros of real polynomial systems. However, unlike~\cite{EPR19}, our analysis applies to structured polynomials that are sparse but with a combinatorial restriction on the support. We note that our sparseness condition is similar to that of~\cite{renegar1987}
and so is the bound we obtain; the latter is polynomial in the degree and the size of the support and exponential in the number of variables. Many computational problems in real algebraic geometry lack algorithms that are polynomial in the degree, so such bounds push the limits of the state-of-the-art. 

In the case of univariate polynomials, our results imply that the complex roots of a random real univariate sparse polynomial around the unit interval are well separated with high probability. 
The logarithm of the separation bound is an important parameter that controls the complexity of many, if not all, univariate solvers. We exploit its relation with the condition number to obtain bounds on the size of the subdivision tree of Descartes' solver and on the average run-time of the sparse univariate solver of \cite{jindalsagraloff2017}.

In both cases, that for both Descartes' solver and the solver of \cite{jindalsagraloff2017}, the bounds that we obtain are: (i) polynomial in the size of the sparse polynomial (meaning polynomial in the size of the support and the logarithm of the degree), and (ii) extend to all higher moments. The importance of these bounds is that they are the first step towards solving the real sparse version of Smale's 17th problem, as stated by \cite{rojasye2005}.

Our framework is based on the one hand on variational properties of the polynomials and the corresponding condition numbers and on the other on probabilistic techniques from geometric functional analysis. The former follows the variational approach to condition numbers of~\cite[2\textsuperscript{\S2}]{tonellicuetothesis} and extends~\cite{CETC-norms} to new norms. The latter has been already applied by~\cite{EPR18,EPR19} and by ~\cite{CETC-PV}, but the way that we apply these methods takes them to the maximum development. 

The $1$-norm on the space of polynomials behaves as the ``dual''norm to the $\infty$-norm on the cube. This norm is naturally suited for subdivision methods on the cube. The analysis of the Plantinga-Vegter subdivision process using our framework serves the purpose to convince the reader of the advantages of the new framework for the analysis of algorithms.
It also has the ambition to bring new insights in the study of algorithms in numerical algebraic geometry.
Our approach continues the trend started by~\cite{CETC-PV} of bringing further interactions between the communities of numerical algebraic geometry and symbolic computation. 

\vspace{10pt}
A preliminary version of the paper appeared in
the proceedings of ISSAC 2020 \citep{TCT-cube-I-20}.
Compared with the conference paper, the current paper
extends significantly the probabilistic model by incorporating
random polynomials whose coefficient are subexponential and
extends significantly the treatment of the univariate case by adding several new results.
Among these results, one can find polynomial (in the size of the support and
logarithm of the degree) bounds for all the moments of the size
of the subdivision tree of Descartes' solver and for the bit complexity of the algorithm by~\cite{jindalsagraloff2017}
for solving a random sparse polynomial. The latter, to our knowledge, is the first such bound.

\paragraph{Notation} We denote by $\Pd$ the space of polynomials in $n$
variables of total degree at most $d$. Then a polynomial is
$f = \sum_{\abs{\alpha} \leq d} f_{\alpha} X^{\alpha} \in \Pd$, where
$\alpha \in \bbN^n$; nevertheless, we commonly omit the summation index. By $\Hd$ we denote the
space of homogeneous polynomials of degree $d$ in $n+1$ variables.

The unit cube is 
$I^n := [-1, 1]^n \subset \bbR^n$ and $B_{\bbC}(x, r)$ is the  complex disk centered at $x$ with radius $r$. The polydisc is $\bbD^n:=\overline{B}_{\bbC}(x, 1)^n\subseteq\bbC^n$.

For $A \subseteq \bbR^n$, we denote by $\mcB(A)$ the set of boxes (i.e., cubes) contained in $A$. For any $B\in\mcB(\bbR^n)$, we denote by $m(B)$ its \emph{midpoint} and by $w(B)$ its \emph{width}, so that $B=m(B)+w(B)/2[-1,1]^n$.

We denote random variables/vectors using letters in fraktur font, such as $\fkc$, $\fkf$ and $\fkx$. We will use $\bbE$ to take expectations. In this way, when the probability measure of $\fkx$ is clear, we will write $\bbE_{\fkx}$ to take the expectation with respect $\fkx$, and we will write $\bbE_{\fkx\in I^n}$ to take the expectation with respect a random $\fkx$ uniformly distributed in $I^n$.

\paragraph{Outline of the paper}

In the next section, we outline and discuss the main results of the paper: the average run-time of the Plantinga-Vegter algorithm,  the average size of the subdivsion tree of Descartes' solver, and the average run-time of the algorithm by~\cite{jindalsagraloff2017}. In Section~\ref{sec:norm-cube}, we introduce the norms with which we will be working and their main properties. In Section~\ref{sec:condition}, we introduce a new condition number adapted to the introduced norms and we prove its main properties.
In Section~\ref{sec:pvalgorithm}, we perform the condition-based complexity analysis of the subdivision routine of the Plantinga-Vegter algorithm;  in Section~\ref{sec:cond-sep}, we introduce the separation bound, give condition-based bounds for it and apply them to the Descartes' solver and the solver by \cite{jindalsagraloff2017}.
Finally, in Section~\ref{sec:prob-estims}, we introduce the randomness model that we will consider, zintzo random polynomials and $p$-zintzo random polynomials and provide the relevant probabilistic bounds to prove our results.

\section{Overview}
\label{sec:main-results}

We present a condition-based framework that allows us to control the
probabilistic analysis of numerical algorithms with respect to random
polynomials that are sparse and do not have any scaling in their coefficients,
as it has been usual with the so-called KSS or dobro random polynomials
introduced by~\cite{CETC-PV}. We illustrate our techniques by analyzing the
expected complexity of the Plantinga-Vegter algorithm and the univariate solvers
\descartes and \JSalg for a class of random sparse polynomials.

We will consider a very general class of random polynomials: the class of zintzo random polynomials (see Definition~\ref{def:zintzo}). Moreover, our probabilistic estimates are both in the average and smoothed complexity  paradigm (see Proposition~\ref{prop:avergaesmoothed}). However, for the sake of concreteness, we will expose our results  only for Gaussian and uniform random polynomials.

\begin{defi}
  \label{def:GU:zintzo}
  Let the finite set $M\subseteq\bbN^n$ be such that it contains $0,e_1,\ldots,e_n$,
  where $e_i$ is the $i$ unit vector in $\bbN^n$. Then
  \begin{enumerate}
    \item[(G)] \label{G:zintzo}
    A \emph{Gaussian polynomial supported on $M$} is a random polynomial $\fkf=\sum_{\alpha\in M}\fkf_\alpha X^\alpha$ supported on $M$ whose coefficients $\fkf_\alpha$ are i.i.d. Gaussian random variables of mean $0$ and variance $1$. 
    \item[(U)] \label{U:zintzo}
    A \emph{uniform random polynomial supported on $M$} is a random polynomial $\fkf=\sum_{\alpha\in M}\fkf_\alpha X^\alpha$ supported on $M$ whose coefficients $\fkf_\alpha$ are i.i.d. uniform random variables on $[-1,1]$. 
  \end{enumerate}
\end{defi}

The condition $\{0,e_1,\ldots,e_n\} \subset M$ is a technicality that we need for the
proofs. In layman's terms, this technical condition states that all the terms
of the first order approximation of $\fkf$ at $0$, that is
$\fkf_0+\fkf_{e_1}X_1+\cdots+\fkf_{e_n}X_n$, appear with probability one.
If we translate this condition to a homogeneous setting, then it holds
\[
  M\subseteq \left\{ (d-1) e_i+e_j\mid i,j\in\{0,\ldots,n\} \right\},
\]
which means that the support contains not only the vertices of the standard simplex, but also the
adjacent lattice points to these vertices. We observe that this sparseness condition, considered already by~\cite{renegar1987}, is a kind of pseudo-sparseness condition. Nevertheless, it is an improvement over the restrictions of other existing analysis such as the one by \cite{EPR19}.

\subsection{Expected complexity of the Plantinga-Vegter algorithm}

The following theorem presents the probabilistic complexity bound for the
subdivision routine of the Plantinga-Vegter algorithm,
\nameref{alg:PVsubdivison}. We refer to Section~\ref{sec:pvalgorithm} for further details on the algorithm.

\begin{theo}\label{cor:mainPVboundgaussian}
  Let $\fkf\in\Pd$ be a random polynomial supported on $M$. The average number
  of boxes of the final subdivision of \nameref{alg:PVsubdivison} using the
  interval approximations~\eqref{eq:intervalapprox1}
  and~\eqref{eq:intervalapprox2} if the input polynomial $\fkf$ is Gaussian
  (Def.~\ref{G:zintzo} (G)) is at most
  \[2n^{\frac{3}{2}}\left(10(n+1) \right)^{n+1}d^{2n}|M|^{n+2} .\] If the input
  polynomial $\fkf$ is uniform (Def.~\ref{U:zintzo} (U)) the bound becomes
  \[2n\,32^{n+1}d^{2n}|M|^{n+2} .\]
\end{theo}

We notice that the previous theorem is a particular case of
Theorem~\ref{thm:mainPVbound}, which gives the claim for the more general class
of zintzo random polynomials.

We notice that the bounds on the number of boxes are polynomial in the degree, as
in~\citep{CETC-PV}. This is an additional theoretical justification of the practical
success of the Plantinga-Vegter algorithm. However, unlike the estimates
in~\citep{CETC-PV}, the bounds we present justify the success of the
Plantinga-Vegter algorithm even for sparse random polynomials. This is one of
the first such probabilistic complexity estimates in numerical algebraic
geometry.

\subsection{Complexity results on univariate solvers}
\label{sec:consep-complexity}

In the setting of univariate solvers, we present two results. First for the
\descartes solver in \cite{sagraloffmehlhorn2016} we bound the average size of
the subdivision tree. Second, we bound the average bit complexity of algorithm
by \cite{jindalsagraloff2017} for isolating the roots of sparse univariate
polynomial. In both cases, we do not only bound the average ($1$st moment) but
all the higher moments.

\begin{theo}\label{theo:descartestreeGaussian}
  Let $\fkf\in\Pd$ be a random polynomial supported on $M$ that is either
  Gaussian or uniform (Def.~\ref{def:GU:zintzo}). The average size of the
  subdivision tree of \descartes on input $\fkf$ is at most
  \[
    \Oh\left(|M|\log d\right).
  \]
  Moreover, the $k$th moment of the size is bounded by
  \[
    \Oh\left(k|M|\log d\right)^k.
  \]
\end{theo}

The above result shows that Descartes' univariate solver can perform well in practice for sparse polynomials in $[-1,1]$. It shows that in average the size of the subdivision tree will be polynomial in the size of the sparse polynomial. This provides an insight on the special character of Mignotte-like 4-nomials. The previous theorem is a particular case of Theorem~\ref{theo:descartestreezintzo}, which states the result for zintzo random polynomials.

\begin{theo}\label{theo:JScomplexityconditionGaussian}
  Let $\fkf\in\Pd$ be a random polynomial supported on $M$ that is either
  Gaussian or uniform (Def.~\ref{def:GU:zintzo}). The average bit-complexity of
  \JSalg on input $(\fkf,I)$ is at most
  \[
    \Oh\left(|M|^{12}\log^{7}d\right).
  \]
  Moreover, the $k$th moment of the bit-run-time is bounded by
  \[
    \Oh\left(k \, |M|^{12}\log^{7}d\right)^k.
  \]
\end{theo}

This theorem provides a first step for the solution of the Rojas-Ye version of
Smale's 17th problem for sparse systems \citep{rojasye2005}. The later theorem
is a particular case of Theorem~\ref{theo:JScomplexityconditionzintzo}, where
the claim is shown for a restricted class of zintzo random polynomials. In
future work, we hope to extend this analysis for polynomials distributed not
only with continuous probability distributions, but also with discrete
probability distributions.

\section{Norms for the cube and polynomials over the cube}
\label{sec:norm-cube}

In the traditional setting, for a homogeneous polynomial
$F=\sum_{\abs{\alpha} = d} F_\alpha X^\alpha\in\Hd$ of degree $d$ in $n+1$
variables, we consider the \emph{Weyl norm},
\[\| F \|_{\mathrm{W}}:=\sqrt{\sum_\alpha\binom{d}{\alpha}^{-1}|F_{\alpha}|^2},\] to
control the evaluations of the polynomial, $F(p)$, and its gradient,
$\nabla_pF$, at points $p\in\bbS^n$.

Unfortunately, the scaling introduced by
the norm in the coefficients affects the probabilistic model and forces us to
consider random polynomials with a particular variance structure that excludes
Gaussian polynomials.
To avoid the scaling of the coefficients, we work in the cube
and we will  use the $\infty$-norm, 
\begin{equation*}
    \|x\|_\infty:=\max_i|x_i|.
\end{equation*}
One of the main disadvantages of this norm is that it does not come from an
inner product.
However, we can overcome this problem as shown by \cite{CETC-norms}.

%
For a polynomial $f:=\sum_{\abs{\alpha} \leq d} f_\alpha X^\alpha\in\Pd$,
motivated by duality,
we  consider the  norm
\begin{equation}
  \label{eq:poly-1-norm}
    \|f\|_1:=\sum_\alpha |f_\alpha| .
  \end{equation}
  
To demonstrate that all the results generalize to the complex case
we will prove the various bounds for polydics,  $z\in \bbD^n:=\overline{B}_{\bbC}(0,1)^n$,
which is the complex analogous of the cube.

The motivation to choose the 1-norm emanates from the following
proposition which shows that we can control the evaluation of $f$ at
$x \in I^n:=[-1,1]^n$, that is $f(x)$, using the 1-norm for $f$.

\begin{prop}\label{prop:l1boundsevaluation}
Let $f\in\Pd$ and $z\in \bbD^n$. Then $|f(z)|\leq \|f\|_1$.
\end{prop}
\begin{proof}
  It holds 
  $|f(z)|=\left|\sum_\alpha f_\alpha z^\alpha\right|
  \leq \sum_\alpha |f_\alpha| \, \|z\|_\infty^{|\alpha|}\leq \|f\|_1$;
  as $z\in D^n$
  implies that $\|z\|_\infty\leq 1$.
\end{proof}

\begin{remark}
  The reader might wonder why we do not choose another norm.
  For example, if we choose $\|f\|_2:=\sqrt{\sum_\alpha |f_\alpha|^2}$, then we
  can prove that for all $z\in D^n$, it holds $|f(z)|\leq \sqrt{N}\|f\|_2$,
  where $N$ is the number of terms in $f$. This gives worse bounds than using
  $\|f\|_1$ since
  \[\|f\|_2\leq \|f\|_1\leq \sqrt{N} \, \|f\|_2,\] which motivates us to prefer
  $\|f\|_1$ to $\sqrt{N}\|f\|_2$ as a bounding quantity.
\end{remark}

\begin{notation}
Before continuing, let us clarify notations so that the statements that we consider are clear. By convention,
\[
  \diff f:= \left(
    \frac{\partial f}{\partial X_1} \, \cdots \, \frac{\partial f}{\partial X_n}
  \right)
\]
represents the formal tangent covector of $f$, whose entries are the formal partial derivatives of $f$. Similarly,  the formal tensor of $k$th derivatives of $f$ is
\[\diff^kf:=\begin{pmatrix}\frac{\partial^kf}{\partial X_{i_1}\, \cdots \, \partial X_{i_k}}\end{pmatrix}_{i_1,\cdots,i_k} .\]

When we want to refer to the gradient vector, the tangent covector or the tensor
of $k$th derivatives of $f$ (evaluated) at a point $z\in\mathbb{R}^n$, we will
use respectively $\diff_zf$ and $\diff_z^kf$. Thus, we have that $\diff_zf$ is a
covector $\Tg_z(\bbR^n)\cong \bbR^n\rightarrow \bbR$ and that $\diff_z^kf$ is a
multilinear map $\Tg_z(\bbR^n)^k\cong (\bbR^n)^k\rightarrow \mathbb{R}$.
\end{notation}

An important feature of the 1-norm is that for polynomials
we can use it to control the  1-norm of their
derivatives. In our notation, if $v_1,\ldots,v_k\in\bbR^n$ and $f\in\Pd$, then $\diff^kf(v_1,\ldots,v_k)$
is a polynomial of degree $\leq d-k$ and so it makes sense to compute its $1$-norm.

\begin{prop}\label{prop:l1boundsderivative}
Let $f\in\Pd$ and $v\in\bbR^n$. Then
\[\|\diff f(v)\|_1\leq d \, \|f\|_1 \, \|v\|_\infty.\]
In particular, for all $z\in \bbD^n$, $|\diff_zf(v)|\leq d \, \|f\|_1 \,\|v\|_\infty$.
\end{prop}
\begin{proof}
We have
\[
d\|f\|_1\|v\|_\infty=\sum_\alpha d|f_\alpha|\|v\|_\infty~\text{ and }~\|\diff f(v)\|_1\leq \sum_\alpha |f_\alpha|\|\diff X^{\alpha}(v)\|_1.
\]
Therefore, it is enough to prove the claim for $X^{\alpha}$. By a simple
computation, we get
\[
\|\diff X^{\alpha}(v)\|_1=\left\|\sum_{i=1}^n\alpha_iv_iX^{\alpha}/X_i\right\|_1\leq \sum_{i=1}^n\alpha_i|v_i|\leq \|\alpha\|_1\|v\|_{\infty}\leq d\|v\|_\infty,
\]
which is the desired inequality for the $1$-norm. The last claim is Proposition~\ref{prop:l1boundsevaluation}.
\end{proof}

\begin{cor}\label{cor:l1boundsderivative}
  Let $f\in\Pd$, $d\geq k\geq 0$, and $v_1,\ldots,v_k\in\bbR^n$. Then
  \[\left\|\frac{1}{k!}\diff^kf(v_1,\ldots,v_k)\right\|_1
    \leq \binom{d}{k} \, \|f\|_1\|v_1\|_\infty\cdots\|v_k\|_\infty.\]
  In particular, for all $z\in D^n$,
  $\left|\frac{1}{k!}\diff^k_zf(v_1,\ldots,v_k)\right|
  \leq \binom{d}{k} \, \|f\|_1\|v_1\|_\infty\cdots\|v_k\|_\infty$.
\end{cor}
\begin{proof}
By induction hypothesis, we have that
\begin{align*}
  \left\|\frac{1}{k!}\diff^kf(v_1,\ldots,v_k)\right\|_1
   & =\frac{1}{k}\left\|\frac{1}{(k-1)!}\diff^{k-1} \left(\diff f(v_k)\right)(v_1,\ldots,v_{k-1})\right\|_1\\
  & \leq \frac{1}{k}\binom{d-1}{k-1}\|\diff f(v_k)\|_1\|v_1\|_\infty\cdots\|v_{k-1}\|_\infty.
\end{align*}
Proposition~\ref{prop:l1boundsderivative} finishes the induction step and provides the base case for induction. The last claim is again Proposition~\ref{prop:l1boundsevaluation}.
\end{proof}

The bounds on the derivatives allows us to bound the (Lipschitz constants) of the
variations of a polynomial $f$ and all its derivatives inside $\bbD^n$.

\begin{prop}\label{prop:l1Lipschitz}
  Let $f\in\Pd$, $d\geq k\geq 0$, and $v_1,\ldots,v_k\in\bbR^n$ such that
  $\|v_i\|_\infty=1$. Then, the map
  \begin{align*}
    \bbD^n&\rightarrow [-1,1]\\
    z&\mapsto \frac{1}{d^k \, \|f\|_1}\diff_z^kf(v_1,\ldots,v_k)
  \end{align*}
  is well-defined and $(d-k)$-Lipschitz with respect the $\infty$-norm.
\end{prop}
\begin{proof}
  Without loss of generality, assume that $\|f\|_1=1$. Let $x,y\in D^n$. By the fundamental theorem of calculus and using the  substitution $t = \tfrac{z-x}{y-x}$,
  we get 
\begin{align*}
  \diff_y^kf(v_1,\ldots,v_k)-\diff_x^kf(v_1,\ldots,v_k)
  & = \int_{x}^{y} \diff_z^{k+1}f(v_1,\ldots,v_k)\,dz\\
  & = \int_0^1\,\diff_{x+ t(y-x)}^{k+1}f(v_1,\ldots,v_k) \cdot (y-x) \,dt.
\end{align*}
Hence, taking absolute values and using Corollary~\ref{cor:l1boundsderivative}, we get that
\[
\left|\frac{1}{d^k}\diff_y^kf(v_1,\ldots,v_k)-\frac{1}{d^k}\diff_x^kf(v_1,\ldots,v_k)\right|\leq \frac{d!}{d^k(d-k-1)!}\|y-x\|_\infty\leq (d-k)\|x-y\|_\infty,
\]
which gives the Lipschitz property. The choice of the co-domain follows from Corollary~\ref{cor:l1boundsderivative}.
\end{proof}

Recall that for a multilinear map $A:(\bbR^n)^k\rightarrow \bbR$ we can consider the induced norm
\[
  \|A\|_{\infty,\infty}:=\sup_{v_1,\ldots,v_k\neq 0}
  \frac{|A(v_1,\ldots,v_k)|}{\|v_1\|_\infty\cdots\|v_k\|_\infty} ,
\]
instead of the $1$-norm
\[\|A\|_1:=\sum_{i_1,\ldots,i_k}|A_{i_1,\ldots,i_k}|. \] Although
$\|A\|_{\infty,\infty}\leq \|A\|_1$ is not an equality in general, it is so in
the case where $A$ is a linear map, which allows us to deduce the following.

\begin{prop}\label{prop:l1Lipschitznorm}
  Let $f\in\Pd$ and $d\geq k\geq 0$. Then the map
  \begin{align*}
    \bbD^n&\rightarrow [0,1]\\
    z&\mapsto \frac{1}{d^k\|f\|_1}\|\diff_z^kf\|_{\infty,\infty}
  \end{align*}
  is well-defined and $(d-k)$-Lipschitz with respect the $\infty$-norm.
\end{prop}
\begin{proof}
Without loss of generality, assume that $\|f\|_1=1$. By Proposition~\ref{prop:l1Lipschitz}, we have that for all $x,y\in \bbD^n$ and all $v_1,\ldots,v_k\in\bbR^n$ such that $\|v_i\|_\infty=1$,
\[\left|\left(\frac{1}{d^k}\diff_y^kf-\frac{1}{d^k}\diff_x^kf\right)(v_1,\ldots,v_k)\right|\leq (d-k)\|y-x\|_{\infty}.\]
By maximizing the left-hand side with respect $v_1,\ldots,v_k$, we get that for all $x,y\in \bbD^n$,
\[
\left\|\frac{1}{d^k}\diff_y^kf-\frac{1}{d^k}\diff_x^kf\right\|_{\infty,\infty}\leq (d-k)\|y-x\|_{\infty}.
\]
This gives the Lipzchitz property. The choice of codomain is justified in a similar way.
\end{proof}

\begin{cor}\label{cor:l1Lipschitznorm}
Let $f\in\Pd$. Then the maps
\begin{align*}
\begin{aligned}
\bbD^n&\rightarrow [0,1]\\
z&\mapsto \frac{1}{\|f\|_1}|f(z)|
\end{aligned}
&&and&&
\begin{aligned}
\bbD^n&\rightarrow [0,1]\\
z&\mapsto \frac{1}{d\|f\|_1}\|\diff_z f\|_{1}
\end{aligned}
\end{align*}
are well-defined and $d$-Lipschitz with respect the $\infty$-norm.
\end{cor}
\begin{proof}
Just note that $\|\diff_zf\|_{\infty,\infty}=\|\diff_zf\|_1$. The rest is straightforward from Proposition~\ref{prop:l1Lipschitznorm}.
\end{proof}

We finish with a slightly stronger version of some of the above results that
will be useful later. When we are outside the polydisk $\bbD^n$, then the
non-linear factors of the polynomials dominate. However, we can retain control
around $\bbD^n$ if we are not too far.

\begin{prop}\label{prop:l1boundsepsilon}
Let $f\in\Pd$, $\varepsilon>0$ and $\bbD^n_{\varepsilon}:=\overline{B}_\bbC(0,1+\varepsilon)^n$. If $\varepsilon\leq\frac{1}{d}$, then:
\begin{enumerate}
\item For all $z\in \bbD^n_{\varepsilon}$, $k\geq 0$, and all $v_1,\ldots,v_k\in\bbR^n$,
\[
  \left|\frac{1}{k!}\diff_z^kf(v_1,\ldots,v_k)\right|
  \leq
  \enumber \, \binom{d}{k} \, \|f\|_1 \, \|v_1\|_\infty\cdots\|v_k\|_\infty.
\]
\item The maps
\begin{align*}
\begin{aligned}
\bbD^n_{\varepsilon}&\rightarrow [0,\enumber]\\
z&\mapsto \frac{1}{\|f\|_1}|f(z)|
\end{aligned}
&&and&&
\begin{aligned}
\bbD^n_{\varepsilon}&\rightarrow [0,\enumber]\\
z&\mapsto \frac{1}{d\|f\|_1}\|\diff_z f\|_{1}
\end{aligned}
\end{align*}
are well-defined and $\enumber d$-Lipschitz with respect the $\infty$-norm.
\end{enumerate}
\end{prop}
\begin{proof}
We consider the polynomial
\[g:=f\left((1+\varepsilon)X\right).\]
We can see that
\[\|g\|_1\leq e\|f\|_1,\]
since for $k\in\{0,1,\ldots,d\}$,
\[(1+\varepsilon)^k\leq \left(1+\frac{1}{d}\right)^{d}\leq \enumber.\]
Moreover, for $z\in D^n_\varepsilon$,
\[
\diff_z^kf=\frac{1}{(1+\varepsilon)^k}\diff_{z/(1+\varepsilon)}^kg,
\]
and so, by Corollary~\ref{cor:l1boundsderivative},
\[
\left|\frac{1}{k!}\diff_z^kf\right|\leq \frac{1}{(1+\varepsilon)^k}\binom{d}{k}\|g\|_1\leq \enumber\binom{d}{k}\|f\|_1.
\]
This proves 1.

The second claim follows from the first one, in the same way  Corollary~\ref{cor:l1Lipschitznorm} follows from Corollary~\ref{cor:l1boundsderivative} (after passing through Propositions~\ref{prop:l1Lipschitz} and~\ref{prop:l1Lipschitznorm}).
\end{proof}

\section{Condition and its properties}
\label{sec:condition}

In this section, we  define and study the properties of the condition number. The following definition adapts the real local condition number~\cite[Chapter~19]{Condition} (cf.~\cite{CKS16}) to our setting.
\begin{defi}
\label{def:l1condition}
  Let $f\in\Pd$ and $x\in I^n$, the \emph{local condition number of
    $f$ at $x$} is the quantity
  \[\cond(f,x):=\frac{\|f\|_1}{\max\left\{|f(x)|,\frac{1}{d}\|\diff_xf\|_1\right\}}.\]
\end{defi}

The intuition behind this condition number is as follows. It holds
$\cond(f,x)=\infty$, if and only if $x$ is a singular zero of $f$. Thus,
$\cond(f,x)$ measures how close is $f$ to have a singularity at $x$. An
important observation is that, since we consider problems from
real algebraic geometry, we do not only have to guarantee that the zeros are smooth,
but also that they exist. The latter is the reason why the term $|f(x)|$ appears
in the denominator and it is  fundamental in numerical real algebraic geometry,
where a perturbation of a polynomial not only perturbs the zeros, but also can
make them disappear.

In~\cite[2\textsuperscript{\S2}]{tonellicuetothesis} a series of explicit pro of the
condition number are underlined as the important properties  to carry out a condition-based complexity analysis. These properties are: the regularity
inequality, the 1st and the 2nd Lipschitz property, and the higher
derivative estimate.
The following theorem shows that the  condition number in \eqref{def:l1condition} has these properties. We recall that Smale's gamma, $\gamma$, is the following invariant\footnote{The formula looks different from the usual one because it is simplified for the case of one multivariate polynomial.}
\begin{equation}
  \label{eq:smalegamma}
  \gamma(f,z):=\sup_{k\geq 2}\left(\frac{1}{\|\diff_zf\|_2}\left\|\frac{1}{k!}\Diff_z^kf\right\|_{2,2}\right)^{\frac{1}{k-1}},
\end{equation}
for a polynomial $f\in\Pd$ and $z\in\bbC^n$, where $\|\cdot\|_{2,2}$ is the induced norm for multilinear maps for the usual Euclidean norm.

\begin{theo}\label{theo:l1conditionproperties}
Let $f\in\Pd$ and $x\in I^n$. Then:
\begin{itemize}
\item \textbf{Regularity inequality}: Either
\[
\frac{|f(x)|}{\|f\|_1}\geq \frac{1}{\cond(f,x)}
\quad~\text{ or }~\quad
\frac{\|\diff_xf\|_1}{d\|f\|_1}\geq \frac{1}{\cond(f,x)}.
 \]
In particular, if $\cond(f,x)\frac{|f(x)|}{\|f\|_1}<1$, then $\diff_x f\neq 0$.
\item \textbf{1st Lipschitz inequality}: The map 
\begin{align*}
    \Pd&\rightarrow [0,\infty)\\
    g&\mapsto \frac{\|g\|_1}{\cond(g,x)}
\end{align*}
is $1$-Lipschitz with respect the $1$-norm. In particular, $\cond(f,x)\geq 1$.
\item \textbf{2nd Lipschitz inequality}: The map 
\begin{align*}
    I^n&\rightarrow [0,1]\\
    y&\mapsto \frac{1}{\cond(f,y)}
\end{align*}
is $d$-Lipschitz with respect the $\infty$-norm. 
\item \textbf{Higher derivative estimate}: If $\cond(f,x)\frac{|f(x)|}{\|f\|_1}<1$, then
\[
\gamma(f,x)\leq  \frac{\sqrt{n}(d-1)}{2}\cond(f,x).
\]
\end{itemize}
\end{theo}
\begin{remark}
  We note that the theorem still holds if we replace  $I^n$ by $D^n$.
\end{remark}

Before continuing with the proof of Theorem~\ref{theo:l1conditionproperties}, let's discuss why these properties are important for us. 
\begin{itemize}
  \item The regularity inequality tells us (in a quantitative way depending on the condition number) that either the value of $f$ at a point $x$ is big or that the gradient of $f$ at that point $x$ is big. In this way, the regularity inequality guarantees us that the covector field $x\mapsto \diff_xf$ does not vanish near the zero set of $f$. The latter allows us to guarantee that the Newton operator is well-defined or that geometric arguments based on following the gradient flow work near the zero set.
  \item The 1st and 2nd Lipschitz properties allow us to guarantee that $\cond(f,x)$ can be numerically evaluated at $(f,x)$, since it guarantees us that $\cond(f,x)$ can be bounded by $\cond(\tilde{f},\tilde{x})$ for sufficiently good approximations $\tilde{f}$ of $f$ and $\tilde{x}$ of $x$. To make this more concrete, we have that
  \begin{equation*}
    \cond(f,x)\leq \frac{\cond\left(\tilde{f},\tilde{x}\right)}{1-\cond\left(\tilde{f},\tilde{x}\right)\left(2\frac{\Onorm*{\tilde{f}-f}}{\Onorm*{\tilde{f}}}+\left\|\tilde{x}-x\right\|_\infty\right)}\leq \cond\left(\tilde{f},\tilde{x}\right)(1+\delta) ,
  \end{equation*}
  whenever $2\frac{\Onorm*{\tilde{f}-f}}{\Onorm*{\tilde{f}}}+\left\|\tilde{x}-x\right\|_\infty <\frac{1}{\cond\left(\tilde{f},\tilde{x}\right)}\delta$, for some $\delta\in(0,1)$.
  \item The higher derivative estimate allows us to control how the Newton method converges near the zero set. It is based on Smale's $\alpha$-theory, for which we refer to~\cite[15.2]{Condition} and \cite[Chapter~4]{dedieubook}, among many other references.
\end{itemize}

\begin{proof}[Proof of Theorem~\ref{theo:l1conditionproperties}]
  For the regularity inequality, following Definition~\ref{def:l1condition},
  $\frac{1}{\cond(f,x)}$ is the maximum of $\frac{|f(x)|}{\|f\|_1}$ and
  $\frac{\|\diff_xf\|_1}{d\|f\|_1}$. Notice that we obtain the two relations
  with equality, but this is not important for the arguments that follow.

For the 1st Lipschitz property, let $g_0,g_1\in\Pd$. Then:
\begin{align*}
    \left|\frac{\|g_0\|_1}{\cond(g_0,x)}\right.&\left.-\frac{\|g_1\|_1}{\cond(g_1,x)}\right|\\
    &=\left|\max\left\{|g_0(x)|,\frac{1}{d}\|\diff_xg_0\|_1\right\}-\max\left\{|g_1(x)|,\frac{1}{d}\|\diff_xg_1\|_1\right\}\right|&\text{(Definition~\ref{def:l1condition}}\\
    &\leq \max\left\{|g_0(x)-g_1(x)|,\frac{1}{d}\|\diff_xg_0-\diff_xg_1\|_1\right\}&\text{(Triangle inequality)}\\
    &\leq \max\left\{|(g_0-g_1)(x)|,\frac{1}{d}\|\diff_x(g_0-g_1)\|_1\right\}&\\
    &\leq \|g_0-g_1\|_1 .&\text{(Proposition~\ref{prop:l1boundsderivative})}
\end{align*}
Note that we use that the codomain is $[0,1]$ in Corollary~\ref{cor:l1Lipschitznorm}. For the second inequality, note that $\frac{\|0\|_1}{\cond(0,x)}=1$ (or simply use (Proposition~\ref{prop:l1boundsderivative} again).

For the 2nd Lipschitz property, the arguments are similar to the ones in the
proof of the 1st Lipschitz property. Without loss of generality, assume that
$\|f\|_1=1$. Let $y_0,y_1\in I^n$, then:
\begin{align*}
    \left|\frac{1}{\cond(f,y_0)}\right.&\left.-\frac{1}{\cond(f,y_1)}\right|\\&=\left|\max\left\{\frac{|f(y_0)|}{\|f\|_1},\frac{\|\diff_{y_0}f\|_1}{d\|f\|_1}\right\}-\max\left\{\frac{|f(y_1)|}{\|f\|_1},\frac{\|\diff_{y_1}f\|_1}{d\|f\|_1}\right\}\right|&\text{(Definition~\ref{def:l1condition}}\\
    &\leq \max\left\{\left|\frac{|f(y_0)|}{\|f\|_1}-\frac{|f(y_1)|}{\|f\|_1}\right|,\left|\frac{\|\diff_{y_0}f\|_1}{d\|f\|_1}-\frac{\|\diff_{y_1}f\|_1}{d\|f\|_1}\right|\right\}&\text{(Triangle inequality)}\\
    &\leq \|y_0-y_1\|_\infty .&\text{(Corollary~\ref{cor:l1Lipschitznorm})}
\end{align*}

For the higher derivative estimate, note that for a multilinear map $A:(\bbR^n)^k\rightarrow\bbR$,
\[
\|A\|_{2,2}=\sup_{x_1,\ldots,x_k\neq 0}\frac{|A(x_1,\ldots,x_k)|}{\|x_1\|_2\cdots\|x_k\|_2}
\leq \sup_{x_1,\ldots,x_k\neq 0}\frac{|A(x_1,\ldots,x_k)|}{\|x_1\|_\infty\cdots\|x_k\|_\infty}=\|A\|_{\infty,\infty},
\]
since $\|z\|_2\geq \|z\|_\infty$, for all $z$.
Also for a linear map $a:\bbR^n\rightarrow \bbR$,
\[
  \|a\|_2\geq \frac{\|a\|_1}{\sqrt{n}}.
\]
In this way, for $k\geq 2$,
\[
\frac{1}{\|\diff_xf\|_2}\left\|\frac{1}{k!}\Diff_x^kf\right\|_{2,2}\leq \frac{\sqrt{n}}{\|\diff_xf\|_1}\left\|\frac{1}{k!}\Diff_x^kf\right\|_{\infty,\infty}=\frac{\sqrt{n}\|f\|_1}{\|\diff_xf\|_1}\frac{\left\|\frac{1}{k!}\Diff_x^kf\right\|_{\infty,\infty}}{\|f\|_1}.
\]
Using the regularity inequality and Corollary~\ref{cor:l1boundsderivative}, the previous inequality becomes
\[
  \frac{1}{\|\diff_xf\|_2}\left\|\frac{1}{k!}\Diff_x^kf\right\|_{2,2}
  \leq \sqrt{n} \, \frac{\cond(f,x)}{d} \,\binom{d}{k}
  = \frac{\sqrt{n}}{d} \binom{d}{k} \, \cond(f,x).
\]
Now, for $k\geq 2$ and since  $\cond(f,x)\geq 1$,
we have $\left|\cond(f,x)\right|^{\frac{1}{k-1}}\leq {\cond(f,x)}$.
Also 
\[
\left|\frac{\sqrt{n}}{d}\binom{d}{k}\right|^{\frac{1}{k-1}}\leq \sqrt{n}\frac{((d-1)\cdots (d-k+1))^{\frac{1}{k-1}}}{(k!)^{\frac{1}{k-1}}}\leq \frac{\sqrt{n}(d-1)}{2},
\]
since $(k!)^{\frac{1}{k-1}}\geq 2$, for $k \geq 2$.
This concludes the proof.
\end{proof}

Using the previous theorem, we can provide a (kind of) geometric interpretation of the condition number (Definition~\ref{def:l1condition}). Fix $x\in I^n$ and consider
\[\Sigma_x:=\{g\in\Pd\mid g(x)=0,\,\nabla_xg=0\}\subset\Pd,\]
which is the subset of polynomials that are singular at $0$. The following
proposition, usually referred as \emph{"Condition Number Theorem"} relates the
distance between a polynomial and $\Sigma_x$ with the condition number. Even
though the version we present is not an equality, it provides a bound in
both directions.

\begin{prop}[Condition Number Theorem]
  \label{prop:l1conditionnumbertheorem}
For all $f\in\Pd$ and $x\in I^n$, 
\[
\frac{\|f\|_1}{\dist_1(f,\Sigma_x)}\leq \cond(f,x)\leq (1+d)\frac{\|f\|_1}{\dist_1(f,\Sigma_x)},
\]
where $\dist_1$ is the distance induced by the $1$-norm.
\end{prop}
\begin{proof}
The left hand side follows from the 1st Lipschitz property (Theorem~\ref{theo:l1conditionproperties}), since for any $g\in\Sigma_x$,
\[
\frac{\|f\|_1}{\cond(f,x)}=\left|\frac{\|f\|_1}{\cond(f,x)}-\frac{\|g\|_1}{\cond(g,x)}\right|\leq \|f-g\|_1.
\]
Thus, by considering the closest $g\in \Sigma_x$ to $f$,
\[
\frac{\|f\|_1}{\cond(f,x)}\leq \dist_1(f,\Sigma_x).
\]

For the right hand side, consider the polynomial
\[g(X) :=f(X) -f(x)-\sum_{i=1}^n\partial_if(x)X_i.\]
Then, for $g\in\Sigma_x$ it holds
$\|f-g\|_1\leq |f(x)|+\|\diff_xf\|_1$. Hence
\[
\dist_1(f,\Sigma_x)\leq \|f-g\|_1\leq (1+d)\max\left\{|f(x)|,\frac{1}{d}\|\diff_xf\|_1\right\}=(1+d)\frac{\|f\|_1}{\cond(f,x)},
\]
as desired.
\end{proof}

We conclude this section, introducing the global condition number and stating its properties. 

\begin{defi}
\label{def:l1globalcondition}
  Let $f\in\Pd$, The \emph{global condition number of
    $f$}
  is 
  \[\cond(f):=\max\{\cond(f,x)\mid x\in I^n\}.\]
\end{defi}

Notice that $\cond(f)$ is infinity if and only if the zero set of $f$ has a
singularity in $I^n$. The following proposition quantifies this fact and
interprets geometrically $f$. We denote by
\[
  \Sigma_{n,d}:=\{g\in\Pd \mid g(x)=0,\,\diff_xg=0 \text{, for some }x\in I^n \}
  =\bigcup_{x\in I^n}\Sigma_x\subset \Pd
\]
the set of polynomials whose zero sets have a singularity in $I^n$.

\begin{prop}\label{prop:l1globalconditionproperty}
  The following  map is $1$-Lipschitz:
  \[f\rightarrow \frac{\|f\|_1}{\cond(f)} .\]
  Moreover, for every $f\in \Pd$,
  \[
    \frac{\|f\|_1}{\dist_1(f,\Sigma_{n,d})}\leq \cond(f)\leq (1+d)\frac{\|f\|_1}{\dist_1(f,\Sigma_{n,d})},
  \]
  where $\dist_1$ is the distance induced by the $1$-norm.
\end{prop}
\begin{proof}
  It follows immediately from the 1st Lipschitz property,
  Theorem~\ref{theo:l1conditionproperties}, and the Condition Number Theorem (Proposition~\ref{prop:l1conditionnumbertheorem}).
\end{proof}

\section{Plantinga-Vegter Algorithm and its complexity}\label{sec:pvalgorithm}

The Plantinga-Vegter algorithm~\citep{plantingavegter2004} is a
subdivision-based algorithm that computes an isotopically correct
approximation of the zeros of a univariate polynomial in an interval,
of a curve in the plane, or of a surface in the $3$-dimensional
space. It can also be generalized to hypersurfaces of arbitrary dimension  \citep{galehousethesis} and to singular curves~\citep{burr2012}.

Following~\cite{burr2017} and~\cite{CETC-PVjournal} (cf. \cite{CETC-PV}), we focus on
the subdivision procedure of the algorithm. To analyze its complexity,
we identify three levels
that we should focus on (following~\cite{xuyap2019} (cf. \cite{yap2019towards})):
\begin{enumerate}
    \item[A)] Abstract level: Evaluations are modelled with exact arithmetic.
    \item[I)] Interval level: Evaluations are modelled with exact interval arithmetic.
    \item[E)] Effective level: Evaluations are modelled with finite precision interval arithmetic.
\end{enumerate}
\cite{CETC-PVjournal} did a complexity analysis of the Plantinga-Vegter algorithm for all three levels. 
However, our objective is not to reproduce this analysis \citep{CETC-PVjournal} to the last detail, but to show that the change from the Weyl norm to the $1$-norm improves the complexity estimates. For this reason, we only focus on the interval level.
The reason for this choice is that if we reproduce all the arguments for the effective level, then the analysis would not only be technically tedious, but will also distract us from highlighting the complexity improvement.

While analyzing the interval level, we  estimate the number of boxes of the final subdivision. This is our measure of complexity. We refer to
~\citet{burr2017},~\citet{CETC-PVjournal}
and~\cite[5\textsuperscript{\S2}]{tonellicuetothesis} for further
justifications of this approach.

\subsection{Interval version of the PV~Algorithm}
\label{sec:interval-PV}

The subdivision routine of the PV algorithm,
\nameref{alg:PVsubdivison}, relies on subdividing the unit cube $I^n$
until each box $B$ in the subdivision satisfies the condition
\[
C_f(B): \text{ either } 0\notin f(B)\text{ or }0\notin \left\{\diff_xf\diff_yf^{\mathrm{T}}\mid x,y\in B\right\}.
\]
Note that $\diff_xf\diff_yf^{\mathrm{T}}=\sum_{i=1}^n\partial_if(x)\partial_if(y)$, as $\diff_zf$ is a covector (and so a row-vector).

\begin{algorithm}[t]
  \SetFuncSty{textsc} \SetKw{RET}{{\sc return}} \SetKw{OUT}{{\sc output \ }} 
	\SetKwInOut{Input}{Input}
	\SetKwInOut{Output}{Output}
    \SetKwInOut{Require}{Require}
    \Input{$f\in \Pd$ which is non-singular in $I^n$}
    \Output{A subdivision $\mcS$ of $I^n$ into boxes\\such that for all $B\in \mcS$, $C_f(B)$ holds}
  
  \BlankLine

  $\mcS_0 \leftarrow \{I^n\}$, $\mcS \leftarrow \varnothing$ \;
  
  \While{ $\mcS_0 \neq \varnothing$}{
    Take $B\in\mcS_0$\;
    \If{$C_f(B)$ holds}{
    $\mcS\leftarrow \mcS\cup\{B\}$, $\mcS_0\leftarrow \mcS_0\setminus \{B\}$\;
    }
    \Else{$\mcS_0\leftarrow \mcS_0\setminus\{B\}\cup\textsc{StandardSubdivsion}(B)$\;}
    }
  \RET $\mcS$ \;
  \caption{\textsc{PV-Subdivsion}}
  \label{alg:PVsubdivison}
\end{algorithm}

To implement this algorithm one uses interval arithmetic. Recall that an \emph{interval approximation} of a map $g:I^n\rightarrow \bbR^q$ is a map $\square [g]:\square[I^n]\rightarrow \square[\bbR^q]$, where $\square[X]$ is the set of (coordinate) boxes contained in $X$, such that for all $B\in\square[I^n]$, we have
\[g(B)\subseteq \square[g](B).\] 

The following proposition provides interval approximations for $f$ and $\Onorm*{\diff f}$.

\begin{prop}\label{prop:intervalapprox}
Let $f\in\Pd$. Then
\begin{equation}\label{eq:intervalapprox1}
   \square[f](B):=f(m(B))+d\|f\|_1\frac{w(B)}{2}[-1,1] 
\end{equation}
and
\begin{equation}\label{eq:intervalapprox2}
    \square[\|\diff f\|_1](B):=\|\diff_{m(B)}f\|_1+\sqrt{2n}d^2\|f\|_1\frac{w(B)}{2}[-1,1]
\end{equation}
are interval approximations of, respectively, $f$ and $\diff f$.
\end{prop}
\begin{proof}
We only need to show, respectively, that $f(B)\subseteq f(m(B))+d\|f\|_1\frac{w(B)}{2}[-1,1] $ and that $\|\diff f\|_1(B)\subseteq \|\nabla_{m(B)}f\|_1+d^2\|f\|_1\frac{w(B)}{2}[-1,1]$. However, this follows  from Corollary~\ref{cor:l1Lipschitznorm}.
\end{proof}

We now show how to test $C_f(B)$ using the above interval approximations. The reason we have the factor $\sqrt{n}$ in~\eqref{eq:intervalapprox2} is so that the next proposition provides a nicer statement.

\begin{prop}\label{prop:intervalcondition}
The condition $C_f(B)$ follows from
\[C_f'(B):\,|f(m(B))|>d\|f\|_1\frac{w(B)}{2}~\text{ or }~\|\diff_{m(B)}f\|_1>d^2\sqrt{2n}\|f\|_1\frac{w(B)}{2}.\]
Hence,~\nameref{alg:PVsubdivison} with the interval approximations given in~\eqref{eq:intervalapprox1} and~\eqref{eq:intervalapprox2} is correct if we substitute the condition $C_f(B)$  by
\[C_f^\square(B):\text{ either }0\notin \square[f](B)\text{ or }0\notin \square[\|\nabla f\|_1](B).\]
\end{prop}
\begin{proof}
On the one hand, by Corollary~\ref{cor:l1Lipschitznorm}, the map $|f|$ is $d\|f\|_1$-Lipschitz. Thus $|f(m(B))|>d\|f\|_1\frac{w(B)}{2}$ implies that for all $x\in B$, $|f(x)|\geq |f(m(B))|-d\|f\|_1\frac{w(B)}{2}$. This is the first clause of $C_f(B)$.

On the other hand, by~\cite[Lemma~4.4]{CETC-PV}, we have that if for all $x\in B$, $\|\diff_xf-\diff_{m(B)}f\|_2\leq\frac{1}{\sqrt{2}}\|\diff_{m(B)}f\|$, then for all $x,y\in B$, $\diff_xf\diff_yf^{\mathrm{T}}\neq 0$, which is the second clause of $C_f(B)$. For $x\in B$,
\[\|\diff_xf-\diff_{m(B)}f\|_2\leq \|\diff_xf-\diff_{m(B)}f\|_1\leq d^2\|f\|_1\frac{w(B)}{2},\]
due to Corollary~\ref{cor:l1Lipschitznorm}. Hence $d^2\|f\|_1\frac{w(B)}{2}\leq \frac{1}{\sqrt{2}}\|\diff_{m(B)}f\|_2$ implies the second clause of $C_f(B)$, and so does $\|\diff_{m(B)}f\|_1\geq d^2\sqrt{2n}\|f\|_1w(B)$, since $\|y\|_1\leq \sqrt{n}\|y\|_2$.

The two paragraphs above together give the desired claim.
\end{proof}

In what follows the interval version of \nameref{alg:PVsubdivison} will be
a variant that exploits
the interval approximations in~\eqref{eq:intervalapprox1} and~\eqref{eq:intervalapprox2}.

\subsection{Complexity analysis of the interval version}

As in~\cite{burr2017} and~\cite{CETC-PV}, our complexity analysis relies on the construction of a local size bound for \nameref{alg:PVsubdivison} and the application of the continuous amortization developed 
by \cite{burr2009,burr2016}.
We recall the definition of the local size bound and the result that
we will exploit in our complexity analysis. Recall that $\bbE_{\fkx\in I^n}$ is the expectation over $\fkx$ uniformly distributed on $I^n$.

\begin{defi}
A \emph{local size bound} for the interval version of \nameref{alg:PVsubdivison} on input $f$ is a function
$b_f:I^n\rightarrow [0,1]$ such that for 
all $x\in\bbR^n$, 
\[
b_f(x)\leq \inf\{\vol(B)\mid x\in B\in \mcB(I^n)
\text{ and }C_f^\square(B)\text{ false}\}.
\] 
\end{defi}

\begin{theo}\label{theo:analysis2}\citep{burr2009,burr2016,burr2017}
The number of boxes of the final subdivision of the interval version of \nameref{alg:PVsubdivison} on input $f$ is at most
\[
4^n\bbE_{\fkx\in I^n}\frac{1}{b_f(\fkx)}.
\]
In addition, the bound is finite if and only if \nameref{alg:PVsubdivison} terminates.\eproof
\end{theo}


\begin{theo}\label{thm:conditionlocalsize}
The function
\[x\mapsto \left(d\sqrt{2n}\cond(f,x)\right)^{-n}\] is a local
size bound for~\nameref{alg:PVsubdivison} on
input $f$.
\end{theo}
\begin{proof}
Without loss of generality, assume that $\|f\|_1=1$. Let $x\in B\in\mcB(I^n)$. Then $\|m(B)-x\|_\infty\leq w(B)/2$ and so, by Corollary~\ref{cor:l1Lipschitznorm} and the regularity inequality (Theorem~\ref{theo:l1conditionproperties}), we have that either
\begin{equation}
    |f(m(B))|>\frac{1}{\cond(f,x)}-d\frac{w(B)}{2}
\end{equation}
or
\begin{equation}
    \|\diff_{m(B)}f\|_1>d\frac{1}{\cond(f,x)}-d^2\frac{w(B)}{2}.
\end{equation}
Hence, $C_f^\square(B)$ is true as long as either $\cond(f,x)^{-1}\geq dw(B)$, or $\cond(f,x)^{-1}>d\sqrt{2n}w(B)$. The result follows, since $\vol(B)=w(B)^n$.
\end{proof}

Theorem~\ref{theo:analysis2} and~Theorem~\ref{thm:conditionlocalsize} result the following corollary, which is the preamble of one of our results. 

\begin{cor}\label{cor:conditionbasedcomplexity}
The number of boxes of the final subdivision of the interval version of \nameref{alg:PVsubdivison} on input $f$ is at most
\[2^{\frac{5}{2}n}n^{\frac{n}{2}}d^n\bbE_{\fkx\in I^n}\cond(f,\fkx)^n .\]
\end{cor}

Theorem~\ref{thm:mainPVbound} follows now from the corollary above and the following proposition.

\begin{remark}
A similar argument as in the proof of~\cite[Theorem~6.4]{CETC-PV} shows that we can bound the local size bound of~\cite{burr2017} in terms of $1/\cond(f,x)^n$. Since the interval approximation of the analyzed version is simpler, requiring a single evaluation, we only analyze the complexity of this.
\end{remark}

\begin{remark}
We note that we can generalize the previous bound for a cube $[-a,a]^n$ bigger than the unit cube. A straightforward argument might try to consider the polynomial $f(aX)$
which considered on $I^n$ behaves like $f$ inside $[-a,a]^n$. Unfortunately,
\[\left\|f(aX)\right\|_1\leq a^d\|f\|_1,\]
which will complicate things as the bounds would become exponential in the degree. To avoid this, one should reprove Corollary~\ref{cor:l1Lipschitznorm}. The trick is to consider the maps
\[
x\mapsto \frac{|f(x)|}{d\|f\|_1\max\left\{1,\|x\|_\infty^d\right\}}~\text{ and }~x\mapsto \frac{|f(x)|}{d^2\|f\|_1\max\left\{1,\|x\|_\infty^{d-1}\right\}} ,
\]
and prove that they are Lipschitz. Let us demonstrate this approach for the first map. We only have to consider the map as the composition of
\[\partial I^n\ni \begin{pmatrix}x_0\\x\end{pmatrix}\mapsto \frac{x_0^d}{d\|f\|_1}f\left(\frac{x}{x_0}\right)\]
together with
\[\bbR^{1+n}\ni x\mapsto \frac{1}{\max\left\{1,\|x\|_\infty\right\}}\begin{pmatrix}1\\x\end{pmatrix}\in\partial I^n,\]
and since each map is Lipschitz, this is also the case for their composition.
\end{remark}

\section{Condition, separation bounds and univariate solvers}
\label{sec:cond-sep}

In this section we turn our attention to the separation of the roots of a real
univariate polynomial. In general, we are interested in the separation between
the real roots and the separation between the conjugate complex roots, as these
bounds usually affect the complexity of (some) univariate solvers. As we will focus on searching
roots in $I$, we will consider the following separation quantities.

\begin{defi}
Let $f\in\Pdone$. Then we define:
\begin{enumerate}
    \item[(R)] The \emph{real separation of $f$}, $\Delta(f)$, is 
    \[
    \Delta(f):=\min\left\{\left|\zeta-\tilde{\zeta}\right|\mid \zeta,\tilde{\zeta}\in I,\,f\left(\zeta\right)=f\left(\tilde{\zeta}\right)=0\right\},
    \]
    if $f$ has no double roots in $I$ and it is  zero otherwise.
  \item[(C)] Let $\varepsilon\in\left(0,\frac{1}{d}\right)$. The \emph{$\varepsilon$-real separation of $f$}, $\Delta_{\varepsilon}(f)$, is
    \[
    \Delta_{\varepsilon}(f):=\min\left\{\left|\zeta-\tilde{\zeta}\right|\mid \zeta,\tilde{\zeta}\in I_\varepsilon:=\{z\in \bbC\mid \dist(z,I)\leq \varepsilon\},\,f\left(\zeta\right)=f\left(\tilde{\zeta}\right)=0\right\},
  \]
    if $f$ has no double roots in $I$ and it is
    zero otherwise.
\end{enumerate}
\end{defi}

We observe that $\Delta(f)$ gives the minimum distance between two real roots of
$f$ in $I$, while $\Delta_{\varepsilon}(f)$ takes into account how near to the
real line complex roots near $I$ are. The quantity $\Delta(f)$ plays a role
controlling the complexity of univariate solvers that do not depend on the
complex roots, such as the Sturm's solver. On the other hand,
$\Delta_{\varepsilon}(f)$ plays a role controlling the complexity of univariate
solvers that depend on the complex roots near the real line, such as the
Descartes' solver.

Our objective is to give lower bounds on these two quantities in terms of the
condition number and use them for analysing two univariate solvers: \descartes
and \JSalg.

\subsection{Condition-based bounds for separation}

For bounding the real separation, we follow the ideas of \cite{raffalli2014}
which allow us to obtain a bound depending on the square root of the global
condition number. The main idea is to exploit  that between two consecutive roots there is a point where the derivative vanishes and so a point where the Taylor expansion becomes quadratic.

\begin{theo}\label{theo:absrealseparation}
  Let $f\in\Pdone$. Then
  $$\Delta(f)\geq \frac{2\sqrt{2}}{d\sqrt{\cond(f)}}.$$
\end{theo}
\begin{proof}
  Let $\zeta,\tilde{\zeta}\in I$ be the pair of real roots of $f$ in $I$ with
  the minimum distance. By Rolle's theorem, there is $x_0\in I$ between them, such that
  \[
    f'(x_0)=0.
  \]
Without loss of generality assume that $\zeta$ is the root closest to $x_0$, so
that $\left|\zeta-\tilde{\zeta}\right|/2\geq |x_0-\zeta|$. Then by Taylor's
theorem (Lagrange's form),
\[
0=f(\zeta)=f(x_0)+\frac{1}{2}f''(c)(\zeta-x_0)^2 ,
\]
for some $c \in I$ between $\zeta$ and $x_0$. Hence
\[
  |f(x_0)|=\frac{1}{2}|f''(c)||\zeta-x_0|^2 \leq
  \frac{\left|\zeta-\tilde{\zeta}\right|^2}{8}|f''(c)| =
  \frac{\left(\Delta(f)\right)^2}{8}|f''(c)|.
\]
Since $f'(x_0)=0$, we obtain the desired result dividing by $\|f\|_1$ and applying Corollary~\ref{cor:l1boundsderivative}. 
\end{proof}

For the $\varepsilon$-real separation bound, our results are based on \cite[Theorem~3.2 and Theorem~5.1]{Dedieu-sep-97}. The main idea is to use the higher derivative estimate and the fact that the inverse of the Smale's $\gamma$ is Lipschitz.

\begin{theo}\label{theo:relrealseparation}
Let $f\in\Pdone$. Then for all $\varepsilon\in\left(0, \frac{1}{\enumber d\cond(f)}\right)$,
    $$\Delta_{\varepsilon}(f)\geq \frac{1}{12d\cond(f)}.$$
\end{theo}
\begin{proof}
Let $\zeta\in I_{\varepsilon}$ be a complex root. By~\cite[Th\'eor\`eme~91]{dedieubook}, the Newton method converges for any point in $B_\bbC(\zeta,1/(6\gamma(f,\zeta))$, where $\gamma$ is Smale's gamma. Hence, taking $\zeta\in I_\varepsilon$ maximizing $\gamma(f,\zeta)$, we have that
\[
\frac{1}{3\gamma(f,\zeta)}\leq \Delta_{\varepsilon}(f),
\]
since $\frac{1}{3\gamma(f,\zeta)}$ is the distance from $\zeta$ to the rest of the roots of $f$. By~\cite[Lemme~98]{dedieubook},
\[
\gamma(f,\zeta)\leq \frac{\gamma(f,x)}{(1-\gamma(f,x)\varepsilon)(1-4\gamma(f,x)\varepsilon+2\gamma(f,x)^2\varepsilon^2)}
\]
for some $x\in I$ such that $|\zeta-x|\leq\varepsilon$. 

By Proposition~\ref{prop:l1boundsepsilon}, we have that
\[\frac{|f(x)|}{\|f\|_1}\leq \enumber d\varepsilon\leq \frac{1}{\cond(f,x)},\]
where the last inequality is by assumption. Hence, by the above inequality and the higher derivative estimate (Theorem~\ref{theo:l1conditionproperties}),
\[
\gamma(f,\zeta)\leq 4(d-1)\cond(f,x),
\]
since $\gamma(f,x)\varepsilon\leq \frac{1}{2\enumber}$. Hence
\[\frac{1}{3\gamma(f,\zeta)}\geq\frac{1}{24(d-1)\cond(f)},\]
which concludes the proof.
\end{proof}

\subsection{Complexity of univariate solvers}\label{subsec:unisolvers}

By ``univariate solver'', we refer to an algorithm that given a univariate
polynomial $f\in\Pdone$ and an interval $J$ where $f$ has only simple roots, it outputs a set of isolating
intervals
for the roots of $f$ in $J$. The latter means that we are focusing on finding real roots. In our case, we will focus in the case
where the interval is $I$.

In general, an univariate solver will be of the form
\[\Pdone\times \bbN\ni (f,L)\mapsto \left\{(x_i,r_i,n_i)\right\}_{i=0}^t\subseteq I\times (0,2^{-L})\times \bbN\]
where the input is a polynomial $f$ and the natural number $L$
and the output,
$\left\{B_\bbC(x_i,r_i)\right\}_{i=0}^t$, is a disjoint family of complex disks with centers at $I$ and radius at $2^{-L}$, each containing at most $n_i$ roots; that is for all $i$ it holds $\left|f^{-1}(0)\cap B_\bbC(x_i,r_i)\right|\leq n_i$. In the particular case, where it also holds that for all $i$,
\[1\leq\left|f^{-1}(0)\cap B_\bbC(x_i,r_i)\right|= n_i,\]
we say that $\left\{(x_i,r_i,n_i)\right\}_{i=0}^t$ is an \emph{$(I,L)$-covering of $f$}~\cite[Definition~1]{jindalsagraloff2017}.

In our computational model, the input polynomials will be \emph{bitstream polynomials}~\cite[Definition~3.35]{eigenwillig_real_2010},
i.e., we can access the coefficients of the polynomial at any desired precision
and we will consider the bit complexity. We will focus on two solvers: \descartes and \JSalg.

\subsubsection{\texorpdfstring{\descartes}{Descartes} Solver}

\descartes is a prominent representative of the subdivision-based algorithms
for isolating the real roots of polynomials, usually with integer coefficients.
The algorithm is extremely efficient in practice \citep{RouZim:solve:03,hemmer2009experimental}.
We refer to~\cite{eigenwillig_real_2010} for general exposition about \descartes.

For simplicity, we would focus on the size of the subdivision tree of \descartes. Recall that the \emph{subdivision tree} of a subdivision-based algorithm is the tree of intervals, ordered by containment, that the algorithm processes during its execution.

\begin{prop}\label{prop:descartestree}
  Given $f\in\Pdone$ with support $|M|$, the algorithm \descartes computes
  isolating intervals for all the roots of $f$ in $I$ in at most
  $\Oh\left(\log\cond(f)+\log d\right)$ iterations (this is the depth of the
  subdivision tree). In particular, the subdivision tree of \descartes at input
  $(f,I)$ has size
\[
\Oh\left(|M|\left(\log\cond(f)+\log d\right)\right).
\]
\end{prop}
\begin{proof}
  We note that an interval $J\subseteq I$ is a terminal interval for \descartes as long as $B_\bbC(m(J),w(J)/2)$ does not contain roots or $B_\bbC(m(J)+i\sqrt{3}w(J)/6,)$ contains exactly one real root, by the one-circle and two-circles theorems~\cite[Theorem~2]{sagraloffmelhorn2011}. Hence the depth of the subdivision tree is at most $\Oh\left(\log\max\{\varepsilon^{-1},\Delta_\varepsilon(f)^-1\}\right)$ for some $\varepsilon>0$, and so, by Theorem~\ref{theo:relrealseparation}, at most
  $\Oh\left(\left(\log\cond(f)+\log d\right)\right)$.

  Now, since $f$ has support $|M|$, the number of sign variations in $(0,\infty)$ and $(-\infty,0)$ is at most $\Oh(|M|)$. Hence, no level of the subdivision tree can have more than $\Oh(|M|)$ nodes by the Schoenberg's theorem~\cite[Theorem~3]{sagraloffmelhorn2011}. This finishes the proof.
\end{proof}

\subsubsection{\texorpdfstring{\JSalg}{JindalSagraloff} Solver}

\cite{jindalsagraloff2017} propose an algorithm, \JSalg, to solve sparse
polynomials. In this setting , the representation of the polynomial $f\in\Pdone$
consists of its support $M\subset\bbN$, which has size $|M|\log d$, and a
sequence of $|M|$ coefficients.

\begin{theo}
The algorithm \JSalg outputs an $(L,I)$-covering of $f$ on input $(f,L)\in \Pdone\times\bbN$ with $f$ supported on $M\subseteq \bbN$. The algorithm \JSalg runs with at most 
\[
\Oh\left(|M|^{12}\log^{3}d\max\{\log^2\|f\|_1,L^{2}\}\right)
\]
bit operations.
\end{theo}
\begin{proof}
This is essentially \cite[Lemma~8]{jindalsagraloff2017} rewritten. For the rest of the claims, one has just to read the assumptions in the introduction.
\end{proof}
\begin{remark}
We note that for dense polynomials the estimate of \JSalg is far from optimal, because of that we will state many of the later results only as $\mathrm{poly}(|M|,\log d)$, since determining the exact form of the polynomial does not lead to optimal bounds.
\end{remark}

The following proposition it is based on the fact that in order to find the roots of a polynomial with an $(I,L)$-cover, we only need to compute $(I,k)$-covers until $k>\max\{\log \Delta_{\varepsilon}(f),\log\varepsilon^{-1}\}$ for some $\varepsilon>0$.
\begin{prop}\label{prop:JScomplexitycondition}
Given $f\in\Pdone$ with support $|M|$, the algorithm \JSalg computes isolating intervals for all the roots of $f$ in $I$ with a run-time of
\[
\Oh\left(|M|^{12}\log^{3}d\max\{\log^2\|f\|_1,\log^3\cond(f)\}\right)
\]
bit operations.
\end{prop}
\begin{proof}
We only need to apply the bounds in Theorem~\ref{theo:relrealseparation}.
\end{proof}

\section{Probability estimates of the condition number and complexity}
\label{sec:prob-estims}

We refine the techniques of~\cite{CETC-PV} to obtain explicit constants in the
bounds and to deal with a restricted classes of sparse polynomials. We also add
certain variations of the randomness models we consider. The relaxations of the
hypotheses are of special interest.

\subsection{Probabilistic concepts and toolbox}
\label{sec:prob-toolbox}

We introduce the relevant probabilistic definitions and results. There will be
two kind of probabilistic results: tail bounds, mainly subgaussian and
subexponential; and anti-concentration bounds.

\subsubsection{Subgaussian and other tail bounds}

The more usual tail bound is some form of Markov's inequality~\cite[Proposition~1.2.4]{V} where for a random variable $\fkx\in\bbR$,
\[
\bbP(\|\fkx\|\geq t)\leq \frac{\bbE|\fkx|^k}{t^k}.
\]
The variables in which we will focus  satisfy such a tail bound but in an stronger sense.

\begin{defi}\cite[see][2.5 and 2.7]{V}
Let $\fkx\in\bbR$ be a random variable. 
\begin{enumerate}
\item[(SE)] We call $\fkx$ \emph{subexponential}, if there
exists an $E>0$ such that for all $t\geq E$,
  \[\bbP(|\fkx|>t)\leq 2\exp(-t /E ).\]
  The smallest
  such $E$ is the \emph{subexponential constant} of $\fkx$.
\item[(SG)] We call $\fkx$ \emph{subgaussian}, if there
  exist a $K>0$ such that for all $t\geq K$,  
  \[\bbP(|\fkx|>t)\leq 2\exp(-t^2/K^2).\]
  The smallest
  such $K$ is the \emph{subgaussian constant} of $\fkx$.
\item[($p$SE)] Let $p\geq 1$. We call $\fkx$ \emph{$p$-subexponential}, if there
  exist a $L>0$ such that for all $t\geq L$,  
  \[\bbP(|\fkx|>t)\leq 2\exp(-t^p/L^p).\]
  The smallest
  such $L$ is the \emph{$p$-subexponential} of $\fkx$.
\end{enumerate}
\end{defi}
\begin{remark}
  The technical term for the subexponetial, subgaussian and $p$-subexponential
  constants are $\psi_1$-, $\psi_2$- and $\psi_p$-norms \cite[see][2.5 and
  2.7]{V}, respectively. However, notice that we can define the norms in many
  ways, even though they are equivalent up to an absolute constant.
\end{remark}
\begin{remark}
  We are not requiring the random variables to be centered, i.e., to have zero expectation. This plays a role into having an uniform approach for the average and smoothed analyses.
\end{remark}

The main inequality that we will use in our bounds of the condition number is a
variant of Hoeffding's inequality.

\begin{prop}\label{prop:normbound}
Let $\fkx\in\bbR^M$ be a random vector.
\begin{enumerate}
  \item[(SE)] If for each $\alpha\in M$, $\fkx_\alpha$ is subexponential with subexponential constant $E_\alpha$, then for all
  $t\geq \sum_{\alpha}E_\alpha$, we have 
  \[
    \bbP(\|\fkx\|_1\geq t) \leq
    2|M|\exp\left(-t/\left(\sum_{\alpha\in M} E_\alpha\right)\right).
  \]
\item[(SG)] If for each $\alpha\in M$, $\fkx_\alpha$ is subgaussian with subgaussian constant $K_\alpha$, then  for all
  $t\geq \sum_{\alpha}K_\alpha$, we have 
  \[
    \bbP(\|\fkx\|_1\geq t) \leq
    2|M|\exp\left(-t^2/\left(\sum_{\alpha\in M}K_\alpha\right)^2\right).
  \]
  \item[($p$SE)] Let $p\geq 1$. If for each $\alpha\in M$, $\fkx_\alpha$ is
    $p$-subexponential with subexponential constant $L_\alpha$, then for all
    $t\geq \sum_{\alpha}L_\alpha$, we have
    \[
      \bbP(\|\fkx\|_1\geq t) \leq
      2|M|\exp\left(-t^p/\left(\sum_{\alpha\in M}L_\alpha\right)^p\right).
    \]
\end{enumerate}
\end{prop}
\begin{proof}
  We prove the last claim as the other two are  particular cases.
  We have that for $t\geq\sum_{\alpha\in M}L_\alpha$,
  \begin{align*}
      \bbP\left(\textstyle\sum_{\alpha\in M}|\fkx_{\alpha}|\geq t\right)&= \bbP\left(\textstyle\sum_{\alpha\in M}|\fkx_{\alpha}|\geq \sum_{\alpha\in M} L_\alpha t/\left(\textstyle\sum_{\alpha\in M}L_\alpha\right)\right)&\left(\sum_{\alpha\in M} L_\alpha t/\left(\textstyle\sum_{\alpha\in M}L_\alpha\right)=t\right)\\
      &\leq \bbP\left(\exists\alpha\in M\mid |\fkx_\alpha|\geq L_\alpha t/\left(\textstyle\sum_{\alpha\in M}L_\alpha\right)\right)&\text{(Implication bound)}\\
      &\leq |M|\max_{\alpha\in M} \bbP\left(|\fkx_\alpha|\geq L_\alpha t/\left(\textstyle\sum_{\alpha\in M}L_\alpha\right)\right)&\text{(Union bound)}\\
      &\leq 2|M|\exp\left(-t^p/\left(\textstyle\sum_{\alpha\in M}L_\alpha\right)^p\right) . &\text{(Hypothesis)}
  \end{align*}
  The implication bound is based on the fact that
  $\sum_{i=1}^nx_i\geq \sum_{i=1}^ny_i$ implies that for some $i$, $x_i\geq y_i$.
  The application of the hypothesis follows from the fact that
  $L_\alpha t/\left(\textstyle\sum_{\alpha\in M}L_\alpha\right)\geq L_\alpha$
  for $t\geq \sum_{\alpha\in M}L_\alpha$.
\end{proof}

\subsubsection{Anti-concentration bounds}

A random variable $\fkx\in\bbR$ such that for some $x\in\bbR$ has
$\bbP(\fkx=x)>0$ is said to be concentrated at $x$. If this is the case with
some of the coefficients of our random polynomial, then we cannot guarantee that
the random polynomial has finite condition almost-surely; it might happen that
it equals a particular ill-posed polynomial (one with infinite condition) with
non-zero probability. Because of this, we introduce the following.

\begin{defi}
  A random variable $\fkx$ has the \emph{anti-concentration property}, if there
  exists a $\rho>0$, such that for all $\eps>0$,
 \[ \max \{ \PP \left( \abs{\fkx-u} \leq \eps \right)
    \mid u \in \bbR \} \leq 2\rho \eps.\]
The smallest such $\rho$ is the anti-concentration constant of $\fkx$.
\end{defi}

The following proposition characterizes anti-concentration in the presense of a
bounded continuous density. We will use this equivalence without mentioning it.

\begin{prop}
 Let $\fkx\in\bbR$ be a random variable. Then $\fkx$ has the anti-concentration property with anti-concentration constant $\leq\rho$ if and only if $\fkx$ is absolutely continuous with respect the Lebesgue measure with density $\delta_\fkx$ bounded by $\rho$.
\end{prop}
\begin{proof}
This is \cite[Proposition~2.2]{RV-1}. The precise bounds follow immediately once we have shown the equivalence.
\end{proof}

We conclude with the following proposition which is a generalization of \cite[Theorem~1.1]{RV-1} for more general linear maps. The explicit constants are thanks to the work of \cite{grigoris16}.

\begin{prop}\label{prop:projboundcont}
Let $A\in\bbR^{k\times N}$ be a surjective linear map and $\fkx\in\bbR^N$ be a random vector such that the $\fkx_i$'s are independent random variables with densities (with respect the Lebesgue measure) bounded almost everywhere by $\rho$. Then, for all measurable $U\subseteq \bbR^k$,
\[\bbP(A\fkx\in U)\leq \vol(U)\left(\sqrt{2}\rho\right)^{k}/\sqrt{\det AA^*}.\]
\end{prop}
\begin{proof}
  Using SVD, write $A=QSP$ where, $P \in \bbR^{k\times N}$ is an
  orthogonal projection, $S$ a diagonal matrix containing the singular
  values of $A$, and $Q$ an orthogonal matrix.

  By~\cite[Theorem~1.1]{RV-1}, see also \cite[Theorem~1.1]{grigoris16} for
  the explicit constant, we have that $P\fkx\in\bbR^k$ is a random
  vector with density bounded, almost everywhere, by $(\sqrt{2}\rho)^k$.
  Hence
  \[\bbP(A\fkx\in U)=\bbP(P\fkx\in (QS)^{-1}U)\leq \vol\left((QS)^{-1}U\right)(\sqrt{2}\rho)^k.\]
  This suffices to conclude  the proof, since $\vol\left((QS)^{-1}U\right)=\vol(U)/\det(QS)$ and $\det(QS)=\sqrt{\det AA^*}$.
\end{proof}

\subsection{Zintzo random polynomials}
\label{subsec:rand-model}

We introduce a new class of random polynomials; we call them \emph{zintzo}%
\footnote{The word ``zintzo'' is a Basque word that means honest, upright, righteous. We note that this is the same meaning as that of the word ``dobro'' in a certain sense. This coincidence is not by chance, since it shows that dobro and zintzo random polynomials are polynomials of the same kind, but different.}
polynomials. 
The zintzo polynomials have many simililarities with the
dobro random polynomials, introduced by \cite{CETC-PV}. The main difference between the two is in the 
variance structure of the coefficients. 
While dobro polynomials scale the coefficients according to the weights induced by the Weyl norm,
zintzo polynomials do not.
This property makes zintzo random polynomials a more natural model of
random polynomials.  Moreover, it allows us to explicitly include sparseness in the
model of randomness.

\begin{defi}
  \label{def:zintzo}
  Let $M\subseteq \bbN^n$ be a finite set such that $0,e_1,\ldots,e_n\in M$. 
  A \emph{zintzo random polynomial supported on $M$} is
  a random polynomial $\fkf=\sum_{\alpha\in M}\fkf_\alpha X^\alpha\in\Pd$
  such that the coefficients $\fkf_\alpha$ are independent subgaussian random variables with the anti-concentration property.
\end{defi}

\begin{remark}
We recall that the condition on the support is important to guarantee the smoothness of zero and necessary for our technical assumptions. We stress this, as we will be assuming this in all this section.
\end{remark}

The bounds and the complexity estimates involving
zintzo random polynomials that we present in the sequel
depend on (the product of) the following two quantities:
\begin{enumerate}
\item The \emph{subgaussian constant} of $\fkf$ which is
  \begin{equation}
    \label{eq:K-sg}
    K_\fkf:=\sum\nolimits_{\alpha\in M}K_\alpha ,
  \end{equation}
  where $K_\alpha$ is the subgaussian constant of $\fkf_\alpha$.
\item The \emph{anti-concentration constant} of $\fkf$ which is
  \begin{equation}
    \label{eq:rho-ac}
    \rho_\fkf:=\sqrt[n+1]{\rho_0\rho_{e_1} \cdots \rho_{e_n}} ,
  \end{equation}
  where $\rho_0$ is the anti-concentration constant of $\fkf_0$ and
  $\rho_{e_i}$ is the anti-concentration constant of
  $\fkf_{e_i}$, for $i \in [n]$.
\end{enumerate}
\begin{exam}
Both Gaussian and uniform random polynomials (recall Definition~\ref{def:GU:zintzo}) are zintzo random polynomials with
\begin{equation}
    K_\fkf\rho_\fkf\leq \frac{|M|}{2}.
\end{equation}
In general, we should think of zintzo random polynomials as a robust version of Gaussian polynomials.
\end{exam}

The product $K_\fkf\rho_\fkf$ is invariant under multiplication of $\fkf$ by a non-zero
scalar and it admits a lower bound.
\begin{prop}
  \label{prop:lowerboundKrho}
  Let $\fkf$ be a zintzo random polynomial supported on $M$.  Then,
  $K_\fkf\rho_\fkf> (n+1)/4\geq 1/2$.
\end{prop}
\begin{proof}
  Using the positivity of the subgaussian constants, $K_\alpha$, of
  the coefficients of the zintzo polynomial $\fkf$ and the
  arithmetic-geometric inequality,
  $$K_\fkf\rho_\fkf\geq
  (n+1)\sqrt[n+1]{(K_0\rho_0)(K_{e_1}\rho_{e_1})\cdots
    (K_{e_n}\rho_{e_n})}.$$ Hence, it suffices to show that for a
  random variable, say $X$, with subgaussian constant $K$ and
  anti-concentration constant $\rho$, $K\rho\geq 1/4$. Now, by definition,
  \[3K\rho\geq \bbP\left(|X|\leq 1.5K\right)=1-\bbP\left(|X|>1.5K\right)\geq 1-2\exp\left(-2.25\right).\]
  By preforming some calculations,
  we get $K\rho\geq 1/4$ as desired.
\end{proof}

\subsubsection{Smoothed analysis}

Recall that in the context of smoothed analysis, as introduced by~\cite{ST:02},
we study the complexity algorithms when the input polynomial is a fixed
polynomial with a random perturbation. The importance of smoothed analysis lies
in that it explains the behaviour of an algorithm in practice better than the
average case, since in practice we tend to have a fixed input with a
perturbation produced by errors. In our setting, it is natural to consider this
perturbation proportional to the norm of the polynomial. The following
proposition shows that for zintzo random polynomials the average complexity
already includes the smoothed case as a particular case. Because of this, there
is no need to give a smoothed analysis of any of our results.

\begin{prop}\label{prop:avergaesmoothed}
  Let $\fkf$ be a zintzo random polynomial supported on $M$, $f\in\Pd$
  a polynomial supported on $M$, and $\sigma>0$. Then,
  $\fkf_\sigma:=f+\sigma\|f\|_1\fkf$ is a zintzo random polynomial
  supported on $M$ such that
  $K_{\fkf_\sigma}\leq \|f\|_1(1+\sigma K_\fkf)$ and
  $\rho_{\fkf_\sigma}\leq \rho_\fkf/(\sigma\|f\|_1)$. In particular,
  \[ K_{\fkf_\sigma}\rho_{\fkf_\sigma} = (K_\fkf+1/\sigma)\rho_\fkf .\]
\end{prop}

Note that both the worst  and the average case are limit cases of the smoothed case. When the perturbation, $\sigma$, becomes zero, the smoothed case becomes the worst case, and when the perturbation has infinite magnitude, the smoothed case becomes the average case. Because of this, it is not surprising that
\[\lim_{\sigma\to0}K_{\fkf_\sigma}\rho_{\fkf_\sigma}=\infty~\text{ and }~\lim_{\sigma\to\infty}K_{\fkf_\sigma}\rho_{\fkf_\sigma}=K_\fkf\rho_\fkf ,\]
which shows that this is the case.

\begin{proof}[Proof of Proposition~\ref{prop:avergaesmoothed}]
It is enough to show that for $x,s\in\bbR$ and a random variable $\fkx$ with subgaussian constant $K$ and anti-concentration constant $\rho$, $x+s\fkx$ is a random variable with subgaussian constant $\leq |x|+sK$ and anti-concentration constant $\leq \rho/s$. We note that the latter follows directly from the definition, so we only prove the former.

Now, for all $t\geq|x|+sK$, 
\[\bbP(|x+s\fkx|\geq t)\leq \bbP(|\fkx|\geq (t-|x|)/s)\leq 2\exp(-(t-|x|)^2/(sK)^2).\]
We can easily check that $t\geq|x|+sK$ implies
$(t-|x|)/(sK)\geq t/(|x|+sK)$. Hence, the claim follows.
\end{proof}

\subsubsection{\texorpdfstring{$p$}{p}-zintzo random polynomials}

Zintzo random polynomials have subgaussian coefficients, we can relax or tighten up this condition which leads to $p$-zintzo random polynomials. We note that $1$-zintzo polynomials will have subexponential coefficients, extending the results that concern zintzo polynomials further. 

\begin{defi}
  \label{def:pzintzo}
  Let $p\geq 1$ and $M\subseteq \bbN^n$ be a finite set such that $0,e_1,\ldots,e_n\in M$. 
  A \emph{$p$-zintzo random polynomial supported on $M$} is
  a random polynomial $\fkf=\sum_{\alpha\in M}\fkf_\alpha X^\alpha\in\Pd$
  such that the coefficients $\fkf_\alpha$ are independent $p$-subexpoential random variables with the anti-concentration property.

Given a $p$-zintzo random polynomial $\fkf\in\Pd$, we define the following quantities:  
  \begin{enumerate}
\item The \emph{tail constant} of $\fkf$ which is
  \begin{equation}
    \label{eq:L-sg}
    L_\fkf:=\sum\nolimits_{\alpha\in M}L_\alpha ,
  \end{equation}
  where $L_\alpha$ is the $p$-subgaussian constant of $\fkf_\alpha$.
\item The \emph{anti-concentration constant} of $\fkf$ which is
  \begin{equation}
    \label{eq:rho-acp}
    \rho_\fkf:=\sqrt[n+1]{\rho_0\rho_{e_1} \cdots \rho_{e_n}} ,
  \end{equation}
  where $\rho_0$ is the anti-concentration constant of $\fkf_0$ and
  $\rho_{e_i}$ is the anti-concentration constant of
  $\fkf_{e_i}$, for $i \in [n]$.
\end{enumerate}
\end{defi}

\begin{exam}
We note that Gaussian polynomials are $p$-zintzo for $p\in [1,2]$, while uniform random polynomials are $p$-zintzo random polynomials with
\[L_\fkf\rho_\fkf\leq \frac{|M|}{2}m \]
for $p\geq 1$. This is the reason we have slightly better results for uniform random polynomials, as we can take the bound for $p$-zintzo random polynomials when $p\to\infty$.
\end{exam}

We note that we only vary the constants that control the how the tail goes to zero as we goes to infinity. Again, our estimates will depend on $L_\fkf\rho_\fkf$ which has a universal lower bound. The proof is analogous to that of Proposition~\ref{prop:lowerboundKrho}.

\begin{prop}
  \label{prop:lowerboundKrhop}
  Let $p\geq 1$ and $\fkf$ a $p$-zintzo random polynomial supported on $M$.  Then
  $L_\fkf\rho_\fkf> 9(n+1)/50\geq 9/25$.\eproof
\end{prop}

In the same way, a version of Proposition~\ref{prop:avergaesmoothed} holds also for $p$-zintzo random polynomials.

\subsection{Condition of zintzo random polynomials}

The following two theorems are our main probabilistic results.

\begin{theo}\label{thm:condtailestimate}
Let $\fkf\in\Pd$ be a zintzo random polynomial supported on $M$. Then for all $t\geq e$,
\[\bbP(\cond(\fkf,x)\geq t)\leq \sqrt{n}d^n|M|\left(\frac{8K_\fkf\rho_\fkf}{\sqrt{n+1}}\right)^{n+1}\frac{\ln^{\frac{n+1}{2}} t}{t^{n+1}}.\]
\end{theo}
\begin{theo}\label{thm:condtailestimatep}
Let $p\geq 1$ and $\fkf\in\Pd$ a $p$-zintzo random polynomial supported on $M$. Then for all $t\geq e$,
\[\bbP(\cond(\fkf,x)\geq t)\leq \sqrt{n}d^n|M|\left(\frac{8L_\fkf\rho_\fkf}{(n+1)^{1-\frac{1}{p}}}\right)^{n+1}\frac{\ln^{\frac{n+1}{p}} t}{t^{n+1}}.\]
\end{theo}

The main idea of the proof is to use a union bound to control $\bbP(\cond(\fkf,x)\geq t)$ by the sum of $\bbP(\|\fkf\|_1\geq u)$ and $\bbP(\|R_x\fkf\|\leq u/t)$ which are controlled, respectively, by Propositions~\ref{prop:normbound} and~\ref{prop:projboundcont}. The reason that this works is because $\bbP(\|\fkf\|_1\geq t)$ decreases exponentially fast. To apply Proposition~\ref{prop:projboundcont}, we need the following lemma, which we proof at the end of this subsection.

\begin{lem}\label{lem:projectionevaluation}
Let $M\subseteq \bbN^n$ as in Definition~\ref{def:zintzo} and $\Pd(M)$ the set of polynomials in $\Pd$ supported on $M$. Let $R_x:\Pd(M)\rightarrow \bbR^{n+1}$ be the linear map given by
\[R_x:f\mapsto \begin{pmatrix}f(x)&\frac{1}{d}\partial_1f(x)&\cdots&\frac{1}{d}\partial_nf(x)\end{pmatrix}^*,\]
and $S:\Pd(M)\rightarrow \Pd(M)$ be the linear map given by
\[S:f=\sum_{\alpha\in M}f_\alpha X^\alpha\mapsto \sum_{\alpha\in M}\rho_{\alpha}f_\alpha X^\alpha,\]
where $\rho\in \bbR_+^M$.
Then for $\tilde{R}_{x}:=R_xS^{-1}$ we have that
\[\sqrt{\det \tilde{R}_{x}\tilde{R}_{x}^*}\geq \frac{1}{d^n\rho_0\rho_{e_1}\cdots\rho_{e_n}} ,\]
with respect to coordinates induced by the standard monomial basis.
\end{lem}

\begin{proof}[Proof of Theorem~\ref{thm:condtailestimate}]
We write $\cond(\fkf,x)=\|f\|_1/\|R_x\fkf\|$, where $R_x$ is as in Lemma~\ref{lem:projectionevaluation}
and the norm $\|\cdot\|$ in the denominator is $\|y\|=\max\{|y_1|,|y_2|+\cdots+|y_{n+1}|\}$.
By the union bound, we have that for $u,s>0$, it holds 
\begin{equation}\label{eq1}
    \bbP(\cond(\fkf,x)\geq t)\leq \bbP(\|\fkf\|_1\geq u)+\bbP(\|R_x\fkf\|\leq u/t).
\end{equation}
By Proposition~\ref{prop:normbound} and for  $u\geq K_\fkf$ it holds
\begin{equation}\label{eq2}
  \bbP(\|\fkf\|\geq u)\leq 2|M|\exp\left(-u^2/K_\fkf^2\right).
\end{equation}
Let $S:\Pd(M)\rightarrow \Pd(M)$ be as in Lemma~\ref{lem:projectionevaluation}
with $\rho_\alpha$ the anti-concentration constant of $\fkf_\alpha$. Then, the
polynomial $S\fkf$ has independent random coefficients with densities bounded
(almost everywhere) by $1$ and so we can apply
Proposition~\ref{prop:projboundcont}. We do so with the help of
Lemma~\ref{lem:projectionevaluation}, so that we obtain
\begin{equation}\label{eq3}
  \bbP(\|R_x\fkf\|\leq u/t)=\bbP(\|\tilde{R}_x(S\fkf)\|\leq u/t)\leq \frac{2^{n+1}}{n!}d^n(\sqrt{2}\rho_\fkf u/t)^{n+1},
\end{equation}
where $\tilde{R}_x$ is as in Lemma~\ref{lem:projectionevaluation}.

Combining \eqref{eq1},~\eqref{eq2}, and~\eqref{eq3} with
$u=K_{\fkf}\sqrt{(n+1)\ln t}$, we get
\[\bbP(\cond(\fkf,x)\geq t)\leq \frac{2|M|}{t^{n+1}}+\frac{2^{n+1}}{n!}d^n\left(\sqrt{2}K_\fkf\rho_\fkf\sqrt{n+1}\right)^{n+1}\frac{\ln^{\frac{n+1}{2}}t}{t^{n+1}}.\]
By Stirling's formula, $(n+1)^{n+1}/n!\leq \sqrt{n}e^{n+1}$, and so the desired claim follows for $t\geq e$, by Proposition~\ref{prop:lowerboundKrho}.
\end{proof}
\begin{proof}[Proof of Theorem~\ref{thm:condtailestimatep}]
The proof is as that of Theorem~\ref{thm:condtailestimate}, but with $p$ instead of $2$ in the exponents of~\eqref{eq2}, $u=L_\fkf\left((n+1)\ln t\right)^{\frac{1}{p}}$ instead of $u=K_\fkf\sqrt{(n+1)\ln t}$, and Proposition~\ref{prop:lowerboundKrhop} instead of Proposition~\ref{prop:lowerboundKrho}.
\end{proof}

\begin{proof}[Proof of Lemma~\ref{lem:projectionevaluation}]
The maximal minor of $A_x$ is given by
\[
  R_X =
  \begin{pmatrix}
    1&x^*\\0&\frac{1}{d}\mathbb{I}
  \end{pmatrix}.
\]

This is precisely the minor associated to the subset
$\{1,X_1,\ldots,X_n\}$ of the standard monomial basis of
$\Pd(M)$. Note that at this point we require the assumption that
$0,e_1,\ldots,e_n\in M$.

By the Cauchy-Binet identity, $\sqrt{\det A_xA_x^*}$ is lower-bounded by the absolute value of the determinant of the given maximal minor. Hence the lemma follows.
\end{proof}

\subsection{Probabilistic complexity analysis for the Plantinga-Vegter algorithm}

We now proceed to the complexity analysis of the condition-based quantities in the complexity analysis of the Plantinga-Vegter algorithm. The theorems of interest are the following two.

\begin{theo}\label{thm:mainPVbound}
  Let $\fkf\in\Pd$ be a zintzo random polynomial supported on $M$. The average number of boxes of the final subdivision of \nameref{alg:PVsubdivison} using the interval approximations~\eqref{eq:intervalapprox1} and~\eqref{eq:intervalapprox2} on input $\fkf$ is at most
\[2n^{\frac{3}{2}}d^{2n}|M|\left(20(n+1) K_\fkf\rho_\fkf\right)^{n+1}.\]
\end{theo}

\begin{theo}\label{thm:mainPVboundp}
  Let $p\geq 1$ and $\fkf\in\Pd$ be a $p$-zintzo random polynomial supported on $M$. The average number of boxes of the final subdivision of \nameref{alg:PVsubdivison} using the interval approximations~\eqref{eq:intervalapprox1} and~\eqref{eq:intervalapprox2} on input $\fkf$ is at most
  \[2n^{\frac{3}{2}}(n+1)^{\frac{1}{p}-\frac{1}{2}}d^{2n}|M|\left(64n^{\frac{1}{p}}(n+1)^{\frac{1}{p}}L_\fkf\rho_\fkf\right)^{n+1}.\]
\end{theo}

The above two results follow from Corollary~\ref{cor:conditionbasedcomplexity} together with the following two results.

\begin{prop}\label{prop:condaverage}
  Let $\fkf\in\Pd$ be a zintzo random polynomial supported on
  $M$. Then, 
\[
\bbE_{\fkf}\bbE_{\fkx\in I^n}\cond(f,x)^n\leq 2n^{2}d^n|M|\left(7\sqrt{n+1} K_\fkf\rho_\fkf\right)^{n+1}.
\]
\end{prop}
\begin{prop}\label{prop:condaveragep}
  Let $p\geq 1$ and $\fkf\in\Pd$ a $p$-zintzo random polynomial supported on
  $M$. Then, 
\[
\bbE_{\fkf}\bbE_{\fkx\in I^n}\cond(f,x)^n\leq 2n^{2}(n+1)^{\frac{1}{p}-\frac{1}{2}}d^n|M|\left(\frac{e^{1-\frac{1}{p}}8}{p^{\frac{1}{p}}}n^{\frac{1}{p}-\frac{1}{2}}(n+1)^{\frac{1}{p}}L_\fkf\rho_\fkf\right)^{n+1}.
\]
\end{prop}
\begin{proof}[Proof of Proposition~\ref{prop:condaverage}]
By the Fubini-Tonelli theorem, we have
\[\bbE_{\fkf}\bbE_{\fkx\in I^n}\cond(f,x)^n=\bbE_{\fkx\in I^n}\bbE_{\fkf}\cond(f,x)^n,\]
so it is enough to compute $\bbE_{\fkf}\cond(f,x)^n=\int_{1}^{\infty}\,\bbP(\cond(\fkf,x)^n\geq t)\,\mathrm{d}t$. The latter, by Theorem~\ref{thm:condtailestimate}, is bounded from above by
\[
e^n\sqrt{n}d^n|M|\left(\frac{8K_\fkf\rho_\fkf}{\sqrt{n}\sqrt{n+1}}\right)^{n+1}\int_{1}^{\infty} \frac{\ln^{\frac{n+1}{2}}t}{t^{1+\frac{1}{n}}}\,\mathrm{d}t.
\]

After straightforward calculations, we obtain
\[
\int_{1}^\infty \frac{\ln^{\frac{n+1}{2}}t}{t^{1+\frac{1}{n}}}\,\mathrm{d}t=n^{\frac{n+3}{2}}\Gamma\left(\frac{n+3}{2}\right)\leq e\sqrt{\pi}n^{\frac{n+4}{2}}\left(\frac{n+1}{2e}\right)^{\frac{n+1}{2}} ,\]
where $\Gamma$ is Euler's Gamma function and the second inequality follows from Striling's approximation. Hence, the bound follows.
\end{proof}
\begin{proof}[Proof of Proposition~\ref{prop:condaveragep}]
We do as in the proof of Proposition~\ref{prop:condaverage}, but applying Theorem~\ref{thm:condtailestimatep} instead of Theorem~\ref{thm:condtailestimate}. We obtain this way that the desired quantity is bounded by
\[
e^n\sqrt{n}d^n|M|\left(\frac{8L_\fkf\rho_\fkf}{\sqrt{n}(n+1)^{1-\frac{1}{p}}}\right)^{n+1}\int_{1}^{\infty} \frac{\ln^{\frac{n+1}{p}}t}{t^{1+\frac{1}{n}}}\,\mathrm{d}t.
\]
Computing the integral turns into
\[
\int_{1}^\infty \frac{\ln^{\frac{n+1}{p}}t}{t^{1+\frac{1}{n}}}\,\mathrm{d}t=n^{1+\frac{n+1}{p}}\Gamma\left(\frac{n+1}{p}+1\right)\leq e\sqrt{\pi}n^{2+\frac{n+1}{p}}\left(\frac{n+1}{pe}\right)^{\frac{n+1}{p}} ,\]
after we apply Stirling's approximation in the last inequality. The result now follows.
\end{proof}

\subsection{Probabilistic complexity analysis for univariate solvers}

For \descartes and \JSalg, we prove the two following general result.
\begin{theo}\label{theo:descartestreezintzo}
Let $p\geq 1$ and $\fkf\in\Pd$ be a zintzo random polynomial supported on $M$. The average size of the subdivision tree of \descartes on input $\fkf$ is at most
\[
\Oh\left(|M|\,|\log (dL_\fkf\rho_\fkf)|\right).
\]
Moreover, the $k$th moment of the size is bounded by
\[
\Oh\left(k|M|\,|\log (dL_\fkf\rho_\fkf)|\right)^k.
\]
\end{theo}

\begin{theo}\label{theo:JScomplexityconditionzintzo}
Let $p\geq 1$ and $\fkf\in\Pd$ be a $p$-zintzo random polynomial supported on $M$. The average bit-complexity of \JSalg on input $(\fkf,I)$ is at most
\[
\Oh\left(|M|^{12}\log^{6}d|\log(L_\fkf\rho_\fkf)|\right).
\]
Moreover, the $k$th moment of the bit-run-time is bounded by
\[
\Oh\left(k|M|^{12}\log^{6}d|\log(L_\fkf\rho_\fkf)|\right)^k.
\]
\end{theo}
\begin{remark}
The global assumption on the anti-concentrarion constants is to control $\log\|\fkf\|_1$ when $\|\fkf\|_1$ is small.
\end{remark}

These results follow from Propositions~\ref{prop:descartestree} and~\ref{prop:JScomplexitycondition} and the following two propositions. We give the computations below.

\begin{prop}
  \label{prop:globalcondbound}  
  Let $\fkf\in\Pd$ be a zintzo random polynomial supported on
  $M$. Then, for all $t>2e$, 
\[\bbP(\cond(\fkf)\geq t)\leq 2\sqrt{n}d^{2n}|M|\left(\frac{16K_\fkf\rho_\fkf}{\sqrt{n+1}}\right)^{n+1}\frac{\ln^{\frac{n+1}{2}} t}{t}\leq  2\sqrt{n}d^{2n}|M|\left(10K_\fkf\rho_\fkf\right)^{n+1}\frac{1}{\sqrt{t}}.\]
\end{prop}
\begin{prop}
  \label{prop:globalcondboundp}  
  Let $p\geq 1$ and $\fkf\in\Pd$ be a $p$-zintzo random polynomial supported on
  $M$. Then, for all $t>2e$, 
\[\bbP(\cond(\fkf)\geq t)\leq 2\sqrt{n}d^{2n}|M|\left(\frac{16L_\fkf\rho_\fkf}{(n+1)^{1-\frac{1}{p}}}\right)^{n+1}\frac{\ln^{\frac{n+1}{p}} t}{t}
\leq 2\sqrt{n}d^{2n}|M|\left(\frac{16L_\fkf\rho_\fkf}{(n+1)^{1-\frac{2}{p}}}\right)^{n+1}\frac{1}{t^{1-\frac{1}{p}}}.\]
\end{prop}
\begin{proof}[Proof of Proposition~\ref{prop:globalcondbound}]
The idea is to use an efficient $\varepsilon$-net of $I^n$ and the
2nd Lipschitz property to turn our local estimates into global ones,
as is done
in~\cite[Theorem~1\textsuperscript{\S2}19]{tonellicuetothesis}. Recall,
that an $\varepsilon$-net of $I^n$ (with respect to the $\infty$-norm)
is a finite subset $\mcG\subseteq I^n$ such that, for all $y\in I^n$,
$\dist_\infty(y,\mcG)\leq \varepsilon$.

Note that for every $\varepsilon\in (0,1)$, we have an $\varepsilon$-net
$\mcG_\varepsilon\subseteq I^n$ of size $\leq
2\varepsilon^{-n}$. To construct it, we  take the uniform grid
in the cube.

If $\cond(\fkf)\geq t$, then
\[\max\left\{\cond(\fkf,x)\mid x\in\mcG_{\frac{1}{dt}}\right\}\geq \frac{t}{2}\]
by the 2nd Lipschitz property (Theorem~\ref{theo:l1conditionproperties}). Effectively, let $x_\ast\in I^n$ such that $\cond(\fkf)=\cond(\fkf,x_\ast)$, then there is $x\in\mcG_{\frac{1}{dt}}$ such that $\dist_\infty(x,x_\ast)\leq \frac{1}{dt}$ and therefore
\[
\max\{\cond(\fkf,x)\mid x\in\mcG_{\frac{1}{2dt}}\}\geq \cond(\fkf,x)=\left(\frac{1}{\cond(\fkf,x)}\right)^{-1}\geq \left(\frac{1}{\cond(\fkf,x_\ast)}+\frac{1}{t}\right)^{-1}=\left(\frac{1}{\cond(\fkf)}+\frac{1}{t}\right)^{-1}\geq \frac{t}{2}.
\]
To obtain the first inequality, we argue as follows:
\begin{align*}
\bbP(\cond(\fkf)\geq t)&\leq \bbP\left(\exists x\in\mcG_{\frac{1}{dt}}\mid \cond(\fkf,x)\geq \frac{t}{2}\right)&\text{(Implication bound)}\\
&\leq \left|\mcG_{\frac{1}{dt}}\right|\max\left\{\bbP\left(\cond(\fkf,x)\geq \frac{t}{2}\right)\mid x\in\mcG_{\frac{1}{dt}}\right\}&\text{(Union bound)}\\
&\leq 2d^nt^n\max\left\{\bbP\left(\cond(\fkf,x)\geq \frac{t}{2}\right)\mid x\in\mcG_{\frac{1}{dt}}\right\}&\left(\left|\mcG_{\frac{1}{dt}}\right|\leq 2d^nt^n\right)\\
&\leq 2\sqrt{n}d^{2n}|M|\left(\frac{16K_\fkf\rho_\fkf}{\sqrt{n+1}}\right)^{n+1}\frac{\ln^{\frac{n+1}{2}} t}{t}&\text{(Theorem~\ref{thm:condtailestimate})}
\end{align*}
For the second one, note that $\ln^{n+1}t\leq \frac{(n+1)^(n+1)}{e^{n+1}}t$
\end{proof}
\begin{proof}[Proof of Proposition~\ref{prop:globalcondbound}]
Like the proof of Proposition~\ref{prop:globalcondbound}, but using Theorem~\ref{thm:condtailestimate} instead of Theorem~\ref{thm:condtailestimatep}.
\end{proof}

We restrict ourselves now to the technical results needed in the proofs. The following propistion proves Theorem~\ref{theo:descartestreezintzo} and so Theorem~\ref{theo:descartestreeGaussian}.

\begin{prop}
  \label{prop:globalcondboundlogp}  
  Let $p\geq 1$ and $\fkf\in\Pdone$ be a $p$-zintzo random polynomial supported on
  $M$. Then for $k\geq 1$,
  \[
  \bbE_\fkf\log^k\cond(\fkf)\leq \Oh\left(k(\log(dL_\fkf\rho_\fkf))\right)^k.
  \]
\end{prop}
\begin{proof}
Let $\fkx\in[1,\infty)$ be a random variable such that
\[
\bbP(\fkx\geq t)\leq Ct^{-\alpha}
\]
for all $t\geq t_0$, where $C,\alpha>0$ and $t_0\geq 2$. Then we have that
\[\bbE_\fkx \log^k\fkx\leq \log^k t_0+\frac{Ck^k}{\alpha^kt_0^{\frac{\alpha}{2}}}.\]

To see this note that
\[
\int_{0}^\infty\,\bbP\left(\log^k\fkx\geq s\right)\,\mathrm{d}s\leq \log^kt_0+\int_{\log^kt_0}^\infty\,\bbP\left(\log^k\fkx\geq s\right)\,\mathrm{d}s\leq\log^kt_0+  C\int_{\log^k t_0}^\infty\,2^{-\alpha s^{\frac{1}{k}}}\,\mathrm{d}s.
\]
Then, we have that
\begin{align*}
    \int_{\log^k t_0}^\infty\,2^{-\alpha s^{\frac{1}{k}}}\,\mathrm{d}s
    &=\frac{k}{\alpha^k}\int_{\alpha \ln 2\log t_0}^\infty u^{k-1}e^{-u}\,\mathrm{d}u&\left(s=\alpha \ln 2\,t^{\frac{1}{k}}\right)\\
    &\leq \frac{k}{\alpha^k}\left(\frac{2(k-1)}{e}\right)^{k-1}\int_{\alpha \ln 2\log t_0}^\infty e^{-\frac{u}{2}}\,\mathrm{d}u&\left(u^{k-1}e^{-\frac{u}{2}}\leq \left(\frac{2(k-1)}{e}\right)^{k-1}\right)\\
    &=\frac{k}{\alpha^k}\left(\frac{2(k-1)}{e}\right)^{k-1}t_0^{-\frac{\alpha}{2}}\\
    &\leq \frac{k^k}{\alpha^kt_0^{\frac{\alpha}{2}}}.
\end{align*}

Now, we take $\log t_0=\Omega\left(\log d+\log (L_\fkf\rho_\fkf)\right)$.
\end{proof}

The following two propositions gives the proof of Theorem~\ref{theo:JScomplexityconditionzintzo} and so of Theorem~\ref{theo:JScomplexityconditionGaussian}. For applying it, just notice that
\[\bbE\max\{|\fkx|,|\fky|\}\leq \bbE|\fkx|+\bbE|\fky|.\]

\begin{prop}
  \label{prop:lognormboundp}  
  Let $p\geq 1$ and $\fkf\in\Pdone$ be a $p$-zintzo random polynomial supported on
  $M$. Assume that for all $\alpha\in M$, the anti-concentration constant of $\fkf_\alpha$ is $\leq 1$. Then for $k\geq 1$,
  \[
  \bbE_\fkf|\log^k\|f\|_1|\leq \Oh\left(k^k+\left(1+\frac{k^k}{p^k}\right)\log^k L_\fkf\right).
  \]
\end{prop}

\begin{proof}
We have that
\[
\bbE_\fkf|\log^k\|f\|_1|=\int_{0}^{1} |\log^k t|\delta_{\|f\|_1}(t)\,\mathrm{d}t+\int_{1}^{\infty} |\log^k t|\delta_{\|f\|_1}(t)\,\mathrm{d}t,
\]
where $\delta_{\|\fkf\|_1}$ is the density of $\fkf$. 

By our assumption, $\fkf$ has density (with respect the Lebesgue measure) and this density is bounded by $1$. Hence we have that for $J\subseteq [0,1]$,
\[\bbP(\|\fkf\|_1\in J)=\vol\{x\in\bbR^M\mid \|x\|_1\in J\}\leq 2w(J);\]
and so $\|\fkf\|_1$ has density bounded by $2$. To finish the estimation of the first summand, note now that
\[
\int_{0}^{1} |\log^k t|\,\mathrm{d}t=\int_{0}^\infty\,e^{-s^{\frac{1}{k}}}\,\mathrm{d}s=k!.
\]

If we define the random variable $\fkx$ to be $\|f\|_1$ if $\|f\|_1\geq 1$ and zero otherwise, then
\[
\int_{1}^{\infty} |\log^k t|\delta_{\|f\|_1}(t)\,\mathrm{d}t=\bbE_\fkx \log^k\fkx
\]
and $\fkx$ satisfies
\[
\bbP(\fkx\geq t)\leq e^{-\frac{t^p}{L_\fkf^p}}\leq \frac{L_\fkf^p}{t^p}
\]
for $t\geq L_\fkf$, by Proposition~\ref{prop:normbound}. Arguing as in the proof of Proposition~\ref{prop:globalcondboundlogp}, we obtain then that
\[
\bbE_\fkx \log^k\fkx\leq \left(1+\frac{k^k}{p^k}\right)\log^k L_\fkf.
\]

\end{proof}


\section*{Acknowledgments}
The first author is supported by a postdoctoral fellowship of the 2020 ``Interaction'' program of the Fondation Sciences Mathématiques de Paris. The authors are partially supported by ANR JCJC GALOP (ANR-17-CE40-0009), the PGMO grant ALMA and the PHC GRAPE.

Both authors are grateful to Alperen Ergür for various discussions and suggestions. The first author is grateful to Evgenia Lagoda for moral support and Gato Suchen for useful discussions.

\bibliographystyle{elsarticle-harv}
\bibliography{biblio}

\end{document}